\documentclass[a4paper,fleqn]{cas-sc}

\usepackage[numbers]{natbib}

\usepackage[T1]{fontenc}
\usepackage{graphicx}
\usepackage[font=small,labelfont=bf]{caption}
\usepackage[font=small,labelfont=small]{subcaption}
\usepackage{amssymb,amsthm}
\usepackage[capitalise]{cleveref}
\usepackage{thm-restate}
\usepackage{apptools}
\usepackage{enumitem}
\usepackage{afterpage}

\graphicspath{{figures/}}

\newtheorem{theorem}{Theorem}

\newtheorem{proposition}[theorem]{Proposition}
\newtheorem{observation}[theorem]{Observation}
\newdefinition{remark}{Remark}

\usepackage{xcolor}
\definecolor{defblue}{rgb}{0.121,0.47,0.705}
\let\emph\relax
\DeclareTextFontCommand{\emph}{\color{defblue}\em}

\begin{document}
\let\WriteBookmarks\relax
\renewcommand{\topfraction}{1} %
\renewcommand{\bottomfraction}{1} %
\renewcommand{\textfraction}{0.01} %
\renewcommand{\floatpagefraction}{1} %

\shorttitle{Morphing Graph Drawings
	in~the~Presence~of~Point~Obstacles}    

\shortauthors{Firman, Hegemann, Klemz, Klesen, Sieper, Wolff, and Zink}  

\title [mode = title]{Morphing Graph Drawings
	in~the~Presence~of~Point~Obstacles}  

\tnotemark[1] 

\tnotetext[1]{Work partially supported by DFG grants WO~758/9-1 and WO~758/11-1.
	A preliminary version of this article appeared in the proceedings of
	SOFSEM 2024~\cite{fhkkswz-mgdppo-sofsem24}.} 

\author[1]{Oksana Firman}[orcid=0000-0002-9450-7640]
\author[1]{Tim Hegemann}[orcid=0009-0008-4770-3391]
\author[1]{Boris Klemz}[orcid=0000-0002-4532-3765]
\cormark[1]
\author[1]{Felix Klesen}[orcid=0000-0003-1136-5673]
\author[1]{Marie Diana Sieper}[orcid=0009-0003-7491-2811]
\author[1]{Alexander Wolff}[orcid=0000-0001-5872-718X]
\author[1]{Johannes Zink}[orcid=0000-0002-7398-718X]

\affiliation[1]{organization={Institut für Informatik, Universität
    Würzburg}, city={Würzburg}, country={Germany}}

\cortext[1]{Corresponding Author}
\cortext[0]{\includegraphics[height=8pt]{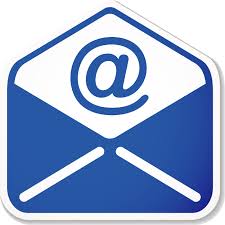}~\texttt{firstname.lastname@uni-wuerzburg.de}}

\begin{abstract}
A crossing-free morph is a continuous deformation between two graph drawings that preserves straight-line pairwise noncrossing edges.
Motivated by applications in 3D morphing problems, we initiate the study of morphing graph drawings in the plane in the presence of stationary point obstacles, which need to be avoided throughout the deformation.

As our main result, we prove that it is NP-hard to decide whether such
an obstacle-avoiding 2D morph between two given drawings of the same graph exists.
In fact, this statement remains true even in the severely restricted special case where only three vertices have to change positions.
This is in sharp contrast to
the classical case without obstacles, where there is an efficiently
verifiable (necessary and sufficient) criterion for the existence of a
morph.

Further, we provide several combinatorial results related to
conditions under which the existence of a morph between two drawings
of a graph can or cannot be prevented by the placement of a given
number of point obstacles.
\end{abstract}


\maketitle

\section{Introduction}
\label{sec:intro}

In the field of Graph Drawing,
a \emph{morph} between two straight-line
drawings~$\Gamma_1$ and~$\Gamma_2$ of the same graph
is a continuous deformation that transforms~$\Gamma_1$ into~$\Gamma_2$
while preserving straight-line edges at all times.
A morph is \emph{crossing-free}\footnote{In the context of morphs restricted to the
plane $\mathbb R^2$, crossing-free morphs are usually called \emph{planar}.
Since we are also discussing work related to morphs in three dimensions, we are
using the more general term crossing-free in the beginning of our introduction.}
if the edges are pairwise noncrossing at all times.
Morphing (beyond the above, strict definition in
Graph Drawing) has applications in animation and
computer graphics~\cite{gomes1999warping}.
In this paper, we initiate the study of morphing graph drawings in the
presence of stationary point obstacles, which need to be avoided
throughout the motion.

\paragraph{Related work.}
An obvious necessary condition for the existence of a crossing-free morph in~$\mathbb R^2$ between two straight-line drawings~$\Gamma_1$ and~$\Gamma_2$ is that these drawings represent the same \emph{plane graph}
(i.e., a planar graph equipped with fixed combinatorial embedding and a distinguished outer face).
It has been established long ago~\cite{Cairns,Thomassen}
that this (efficiently verifiable) criterion is also sufficient, i.e.,
a crossing-free morph in~$\mathbb R^2$ between two straight-line drawings of the same plane graph always {\em exists}.

More recent work~\cite{DBLP:journals/siamcomp/AlamdariABCLBFH17,DBLP:conf/esa/Klemz21} focuses on efficient {\em computation} of such morphs.
In particular, this involves producing a discrete description of the continuous motion.
Typically, this is done in form of so-called piecewise linear morphs.
In a \emph{linear} morph, each vertex moves along a straight-line segment at a constant speed (which depends on the length of the segment) 
such that it arrives at its final destination at the end of the morph.
The (unique) linear morph between~$\Gamma_1$ and~$\Gamma_2$ is denoted by~$\langle \Gamma_1,\Gamma_2\rangle$.
A \emph{piecewise linear} morph is created by concatenating several linear morphs, which are referred to as \emph{(morphing) steps}.
A piecewise linear morph consisting of~$k$ steps can be encoded as a sequence of~$k+1$ drawings.
Alamdari et al.~\cite{DBLP:journals/siamcomp/AlamdariABCLBFH17} showed that two straight-line drawings
of the same $n$-vertex plane graph always admit a crossing-free piecewise linear morph in~$\mathbb R^2$
with~$O(n)$ steps, which is best-possible.
Their proof is constructive and corresponds to an $O(n^3)$-time algorithm,
which was later sped up to $O(n^2\log n)$ time by Klemz~\cite{DBLP:conf/esa/Klemz21}.

Other works are concerned with finding crossing-free morphs in~$\mathbb R^2$ between two given drawings while preserving 
certain additional properties, such as convexity~\cite{DBLP:conf/compgeom/AngeliniLFLPR15}, 
upward-planarity~\cite{DBLP:journals/algorithmica/LozzoBFPR20}, or
edge lengths\footnote{In the fixed edge length scenario, the drawings
	are also known as linkages.}
\cite{Connelly,alt2004complexity,connelly2010locked},
or with constructing crossing-free morphs in~$\mathbb R^2$ that transform a given drawing to achieve certain properties,
such as vertex visibilities~\cite{aichholzer2011convexifying} or convexity~\cite{DBLP:journals/comgeo/KleistKLSSS19},
while being in some sense monotonic, in order to preserve the so-called ``mental map''~\cite{DBLP:conf/gd/PurchaseHG06} of the viewer.

Quite recent works~\cite{DBLP:journals/jgaa/ArsenevaBCDDFLT19,agi-mtds3-JGAA23,befklow-mpgdt3d-CGT23}
are concerned with transforming two drawings~$\Gamma_1$ and~$\Gamma_2$ in the plane into each other
by means of crossing-free morphs in the space~$\mathbb R^3$.
Such \emph{2D--3D--2D} morphs are always
possible~\cite{befklow-mpgdt3d-CGT23}---even if~$\Gamma_1$
and~$\Gamma_2$ have different combinatorial embeddings---and
they sometimes require fewer morphing steps than morphs that are restricted to the plane~$\mathbb R^2$~\cite{DBLP:journals/jgaa/ArsenevaBCDDFLT19,agi-mtds3-JGAA23}.
Due to connections to the notoriously open \textsc{Unknot} problem,
3D--3D--3D %
morphs are not well understood and have, so far, only been considered for
trees~\cite{DBLP:journals/jgaa/ArsenevaBCDDFLT19,agi-mtds3-JGAA23}.

\paragraph{Our model and motivation.}
In this paper, we introduce and study a natural variant of the 2-dimensional morphing problem:
given two crossing-free straight-line drawings~$\Gamma_1$ and~$\Gamma_2$
as well as a finite set of points~$P$ in~$\mathbb{R}^2$, called \emph{obstacles}, %
construct (or decide whether there exists) a crossing-free morph in~$\mathbb R^2$ between~$\Gamma_1$ and~$\Gamma_2$ that \emph{avoids}~$P$.
During such a morph, the drawing is not allowed to intersect any of
the obstacles at any time point during the deformation.  The obstacles
remain stationary.

This problem arises naturally when constructing
2D--3D--2D morphs, where it is tempting to apply strategies for the
classical 2-dimensional case on a subdrawing induced by the subset of
the vertices contained in a plane~$\pi$.
Note that every edge between vertices on different sides of~$\pi$
intersects~$\pi$ in a point, which then acts as an obstacle for the
2-dimensional morph.

\paragraph{Conventions and notation.}
In the remainder of the paper, we consider only morphs in the
plane~$\mathbb R^2$.
We write ``drawing'' as a short-hand for ``straight-line drawing in
the plane~$\mathbb R^2$'' and, similarly, we write ``(\emph{planar})
morph'' rather than ``(crossing-free) morph in~$\mathbb R^2$''.  For
any positive integer~$n$, we define $[n]=\{1,2,\dots,n\}$.

\paragraph{Contribution, organization, and further terminology.}
Let $\Gamma_1$ and $\Gamma_2$ be two drawings of the same
plane graph~$G$, and let~$P$ be a set of obstacles.
We say that~$\Gamma_1$ and$~\Gamma_2$ are \emph{blocked} by $P$ if
there is no planar morph between~$\Gamma_1$ and~$\Gamma_2$ that
avoids~$P$ (see \cref{fig:simple-blocked}
for an example).
Moreover, $\Gamma_1$ and $\Gamma_2$ are \emph{blockable} if there
exists a set of obstacles that blocks them.

\begin{figure}[tb]
  \centering
  \includegraphics[page=2]{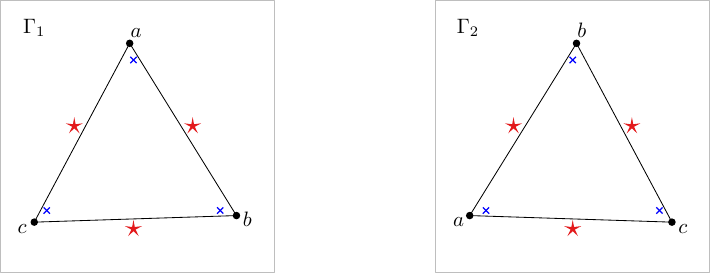}
  \caption{
  Two distinct drawings $\Gamma_1$ and $\Gamma_2$ of a plane cycle $(a,b,c)$ and a set $P$ consisting of three internal
obstacles (blue crosses) and three external obstacles (red stars), which are compatible with $\Gamma_1$ and $\Gamma_2$.
There exists no planar morph between $\Gamma_1$ and $\Gamma_2$ that avoids $P$, i.e., the drawings are blocked by $P$. (For a formal justification, see \cref{prop:block-label-shift-C3-2}.)
}
  \label{fig:simple-blocked}
\end{figure}

We investigate conditions under which drawings can or
cannot be blocked.
In \cref{subsec:not-blockable},
we provide several configurations that cannot be blocked.
We start by observing that drawings of graphs without
cycles (that is, forests) cannot be blocked; see \cref{lem:trees}.
Thus, we concentrate mostly on the case when~$G$ may contain cycles.
\begin{figure}[tb]
	\centering
	\begin{subfigure}{0.45 \linewidth}
		\centering
		\includegraphics[page=2]{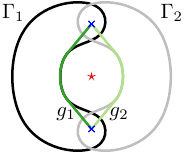}
		\caption{}
		\label{fig:homotopy1}
	\end{subfigure}
	\hfill
	\begin{subfigure}{0.45 \linewidth}
		\centering
		\includegraphics[page=3]{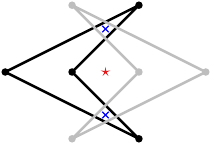}
		\caption{}
		\label{fig:homotopy2}
	\end{subfigure}
	\caption{(a)~Two simple closed curves~$\Gamma_1$
		and~$\Gamma_2$ both containing the two obstacles marked
		by blue crosses and neither containing the obstacle
		marked by a red star.  The curves~$\Gamma_1$
		and~$\Gamma_2$ cannot be continuously deformed into
		each other without passing over one of the obstacles
		and while preserving simplicity.
		(The corresponding continuous deformation of the geodesic
		joining the two blue internal obstacles within the closed
		curve would transform the curve~$g_1$ into~$g_2$ while
		keeping the endpoints fixed and without passing over the
		red external obstacle, which is impossible.)
		Consequently, the two drawings of a 4-cycle
		in~(b) are blocked by the set of three obstacles;
		this set is
		not compatible with the two drawings.}
	\label{fig:homotopy}
\end{figure}
Recall that an obvious necessary condition for the existence of a
planar morph between two drawings is that they represent the same
plane graph.
Interpreting the obstacles in the set~$P$ as (isolated) vertices
reveals that a planar morph between $\Gamma_1$ and $\Gamma_2$ that
avoids $P$ is possible only if each obstacle $p\in P$ is located in
the same face in~$\Gamma_1$ and~$\Gamma_2$.
However, as \cref{fig:homotopy} shows, this condition is not sufficient.
We say that~$P$ is \emph{compatible} with~$\Gamma_1$ and~$\Gamma_2$ if there is a continuous deformation that transforms~$\Gamma_1$ into~$\Gamma_2$
while avoiding~$P$ and preserving pairwise noncrossing ({\em not necessarily straight-line}) edges at all times.
The compatibility of~$P$ with~$\Gamma_1$ and~$\Gamma_2$ is obviously a necessary condition for the existence of a planar obstacle-avoiding morph.
This condition can be checked efficiently~\cite[Theorem~2]{cm14};
note that it is violated in \cref{fig:homotopy2}.
We emphasize that, since the deformation in the definition of compatibility is not
required to maintain straight-line 
edges, our necessary condition is of topological rather than geometric nature.
(We remark that in the field of topology, a plane minus a finite set of obstacle points is often called a ``punctured'' plane. In particular, this term is used in the provided reference~\cite{cm14}.)
We show that two obstacles are never enough to block two
drawings with which the obstacles are compatible, regardless of the
represented graph; see \cref{prop:1or2-obstacles}.

Compatibility is unfortunately not sufficient for the existence
of obstacle-avoiding morphs---even if the considered graph is just a
(3-)cycle.  We study this case in more detail.  Let~\emph{$C_n$} denote the
simple cycle with $n$ vertices.  Let~$\Gamma$ and~$\Gamma'$ be
drawings of a plane~$C_n$ such that (i)~$\Gamma$ and~$\Gamma'$ are
distinct (as mappings of $C_n$ to the plane), but (ii)~the closed
curves realizing~$\Gamma$ and~$\Gamma'$ are identical, and (iii)~the
set of points of~$\mathbb R^2$ used to represent vertices is the same
in~$\Gamma$ and~$\Gamma'$.  Note that there exists an offset
$o \in [n-1]$ such that, for every $i \in [n]$, vertex~$i$ in~$\Gamma$
is at the same location as vertex $i+o$ (modulo $n$) in~$\Gamma'$.
Therefore, we say that~$\Gamma'$ is a \emph{shifted version}
of~$\Gamma$, and we say that $\Gamma$ and $\Gamma'$ are \emph{shifted
  drawings} of~$C_n$.  Due to~(ii), {\em every} set of obstacles is
compatible with~$\Gamma$ and~$\Gamma'$.
We conclude \cref{subsec:not-blockable} by showing that multiple internal {\em and}
external obstacles are required to block drawings of a plane~$C_3$ (\cref{prop:block-label-shift-C3-two-inner,prop:block-label-shift-C3-one-outer})
and tighten these lower bounds on the number of
obstacles needed to block drawings of~$C_3$ in
\cref{subsec:cycle-3} by providing an example of two drawings, which are in fact shifted versions of one another (\cref{prop:block-label-shift-C3-2}).
In \cref{subsec:cycle-n}, we turn our attentions to cycles of arbitrary length.
We show that a set of seven obstacles can block a morph between a
drawing $\Gamma$ of $C_n$ and any shifted version $\Gamma'$ of~$\Gamma$; see
\cref{prop:block-label-shift-in-even-cycles}.
On the other hand, we state a sufficient condition for the existence
of planar obstacle-avoiding morphs between shifted drawings of $C_n$.
We call a degree-2 vertex in a drawing \emph{free} if its two incident
edges lie on a common line.  We show that if there is a free vertex in
a drawing $\Gamma$ of $C_n$, a morph between $\Gamma$ and its shifted
version cannot be blocked (regardless of the number of obstacles); see
\cref{prop:free-vertex-label-shift}.
Free vertices are helpful in other specific cases as well (in
particular, they play a crucial role in the upcoming \cref{thm:NP-hard*}), but their usefulness is limited in general: their
existence is not a sufficient condition for the existence of
obstacle-avoiding morphs even when it comes to (nonshifted drawings of) plane cycles; see
\cref{fig:dense-fox}.
\begin{figure}[tb]
	\centering
	\includegraphics[page=1]{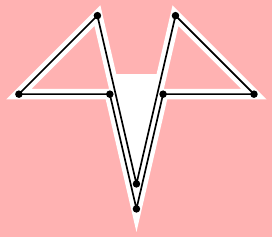}\qquad\includegraphics[page=2]{fox}
	\caption{Two drawings of $C_8$.
		If the (red) shaded regions are densely filled with obstacles, the drawing on the left is essentially locked in place---it cannot be morphed planarly to a substantially different drawing without intersecting the obstacle regions.
		In particular, it cannot be morphed to the drawing on
                the right, even though the right drawing contains two
                free vertices (and the obstacles are compatible with
                the two drawings).}
	\label{fig:dense-fox}
\end{figure}

As our main result, we show that even if the compatibility condition is
satisfied, it is NP-hard to decide whether an obstacle-avoiding planar morph
exists (see \cref{sec:hard}).
In fact, this statement remains true even in the severely restricted special case where only three vertices have to change positions:

\begin{restatable}{theorem}{Hardness}
	\label{thm:NP-hard*}
	Given a plane graph $G$, a set of obstacles $P$,
	and two crossing-free straight-line drawings $\Gamma$ and $\Gamma'$
        of~$G$ in~$\mathbb R^2$, it is NP-hard to
	decide whether there exists an obstacle-avoiding crossing-free morph in $\mathbb R^2$ between~$\Gamma$ and~$\Gamma'$. 
	The problem remains NP-hard when restricted to the case where $G$ is connected,
	the drawings $\Gamma$ and $\Gamma'$ are identical except for the
	positions of three
	vertices, and the obstacles $P$ are compatible
	with $\Gamma$ and $\Gamma'$.  (These statements hold regardless of
	whether the morph is required to be piecewise linear or not.)
\end{restatable}

To prove \cref{thm:NP-hard*}, 
we had to overcome the somewhat intricate challenge of designing gadgets that behave in a synchronous way.
We conclude the paper by discussing several open questions in \cref{sec:conclusion}.

We remark that there is no meaningful difference between piecewise linear morphs and ``general'' morphs~-- as long as one is not interested in bounding the number of morphing steps (which we are not).
Thus, all of our positive results also hold when insisting on piecewise linear morphs.

\begin{observation}\label{obs:piecewise}
If there exists a planar morph between two given drawings~$\Gamma_1$ and~$\Gamma_2$, then
there also exists a piecewise linear planar morph between~$\Gamma_1$ and~$\Gamma_2$
(with possibly many steps).
The same statement is true for obstacle-avoiding planar morphs since the region
in which the drawing is allowed to move is open (i.e., the drawing is
not allowed to intersect the obstacle points).
\end{observation}

\section{Non-Blockable Configurations}
\label{subsec:not-blockable}

In this section, we describe several configurations (combinations of graph classes
and numbers of obstacles) that always allow planar obstacle-avoiding morphs.
First, we show that trees are not blockable regardless of the number of obstacles.

\begin{proposition}
	\label{lem:trees}
	Let $\Gamma_1$ and $\Gamma_2$ be two drawings of the same plane forest~$F$.
	Then $\Gamma_1$ and $\Gamma_2$ are not blockable.
\end{proposition}

\begin{proof}
	We describe how to construct a morph that avoids an arbitrary finite set of obstacles.
	Assume for now that~$F$ consists of a single tree~$T$, which we root at an arbitrary vertex~$r$.
	We construct an obstacle-avoiding planar morph from~$\Gamma_1$ to a drawing~$\Gamma_1'$
	located in a disk that is centered on the position of~$r$ in~$\Gamma_1$,
	contains no obstacles, and whose radius is smaller than the distance~$\gamma$ between any pair of obstacles.
	This can be done by ``contracting'' the tree along its edges
        in a bottom-up fashion (i.e., starting from the leaves); we
        discuss this process in more detail below.
	We also morph~$\Gamma_2$ into an analogously defined drawing~$\Gamma_2'$.
	The drawings can now be translated (far enough) away from the obstacles without intersecting them
	so that they can be transformed into each other by means of
	morphing techniques for the classical non-obstacle
        case~\cite{DBLP:journals/siamcomp/AlamdariABCLBFH17,DBLP:conf/esa/Klemz21}. %
        The combination of these three morphs yields the desired morph from~$\Gamma_1$ to~$\Gamma_2$.
	
	If~$F$ contains multiple trees, it is easy to
        augment~$\Gamma_1$ and~$\Gamma_2$ to drawings of the same
        plane tree by inserting additional vertices and edges, thus,
        reducing to the case of a single tree.
	
	\paragraph{Contraction.}
	It remains to discuss the aforementioned contraction step.
	To this end, it suffices to prove the following statement.
	
	Let~$v$ be a vertex of~$T$ (which is rooted at~$r$), and
        let~$T(v)$ be the subtree of~$T$ rooted at~$v$.
	Further, let~$\Gamma$ be a drawing of~$T$, and let~$\varepsilon$ be a positive number such that the 
	disk~$B(v,\varepsilon)$ of radius~$\varepsilon$ centered on~$v$ contains no obstacle and no vertex 
	other than~$v$ and intersects only edges that are incident to~$v$.
	Then there is a planar obstacle-avoiding morph from~$\Gamma$ to a drawing~$\Gamma^v$ of~$T$ in which
	the subdrawing of~$T(v)$ is located in~$B(v,\varepsilon)$.
        Further, $v$ and all vertices that do not belong to~$T(v)$ are
        placed exactly as in~$\Gamma$.
	
	Note that, indeed, by chosing~$\varepsilon < \gamma$
        and~$\Gamma=\Gamma_1$, the drawing~$\Gamma^r$ corresponds to
        the desired drawing~$\Gamma_1'$.
	
	We prove the statement by induction on the height of~$T(v)$.
	If~$v$ is a leaf, the statement follows trivially.
	So suppose~$v$ is not a leaf.
	Let~$c_1,c_2,\dots,c_k$ be the children of~$v$.
	For every~$i \in [k]$, we choose a radius~$\delta_i\ll
        \varepsilon$ such that the disk~$B(c_i,\delta_i)$ contains
        no vertex other than~$c_i$, intersects only edges incident
        to~$c_i$, and can be translated in
        direction~$\overrightarrow{c_iv}$ until it is fully contained
        in~$B(v,\varepsilon)$ without intersecting any obstacles or
        edges of~$\Gamma$ that are not incident to vertices
        of~$T(c_i)$.
	Moreover, we choose these radii small enough such that when
        carrying out these translations simultaneously, the
        disks~$B(c_i,\delta_i)$ do not intersect each other.
	Clearly, such a set of radii exists.
	Now we can inductively contract each subtree $T(c_i)$ into
        its disk~$B(c_i,\delta_i)$ and simultaneously apply the
        described translations to obtain the desired morph to the
        drawing~$\Gamma^v$.
\end{proof}

Now we show that two obstacles are not enough to block two
drawings (with which the obstacles are compatible), regardless of the
represented graph.
The main idea of the proof is as follows:
we interpret the obstacles as (isolated) vertices of our graph and use
previous results to determine an ``ordinary'' planar morph (in which
the obstacles move).
We then transform this morph into the desired obstacle-avoiding morph by
translating,
rotating, and scaling the frame of reference as time goes on to ensure that
the obstacles become fixed points (one can think of this as moving, rotating,
and zooming
the camera in a suitable fashion).

\begin{proposition}
	\label{prop:1or2-obstacles}
	Let~$\Gamma_1$ and~$\Gamma_2$ be two drawings of the
	same plane graph~$G$, and let~$P$ be a set of obstacles
	that are compatible with~$\Gamma_1$ and~$\Gamma_2$.
	If $|P|\leq 2$, then there exists a planar morph
	from~$\Gamma_1$ to~$\Gamma_2$ that avoids~$P$.
\end{proposition}

\begin{proof}
	We show the statement for a set $P = \{p_1, p_2\}$ of size~2.
	The case $|P| = 1$ can be treated as the case $|P| = 2$ by
        placing a second obstacle in the vicinity of the first one.

	Without loss of generality, we assume that $d(p_1, p_2) = 1$.
	Let~$G'$ be the plane graph obtained from $G$ by introducing two new
	isolated vertices~$v_1$ and~$v_2$ and assigning them into the faces
	where the obstacles $p_1$ and~$p_2$ lie in $\Gamma_1$ (and thus also
	in~$\Gamma_2$).  Let~$\Gamma'_1$ and~$\Gamma'_2$ be the drawings
	of~$G'$ that are obtained from $\Gamma_1$ and $\Gamma_2$,
	respectively, by placing~$v_1$ and~$v_2$ at the positions of~$p_1$
	and~$p_2$, respectively.
	
	We know that there exists a morph~$M_1$ from~$\Gamma'_1$
	to~$\Gamma'_2$~\cite{DBLP:journals/siamcomp/AlamdariABCLBFH17,DBLP:conf/esa/Klemz21}. %
	We can treat $M_1$ as a
	continuous function from the interval $[0,1]$ onto the set of
	straight-line drawings of~$G'$ such that $M_1(0) = \Gamma'_1$ and
	$M_1(1) = \Gamma'_2$.  Since $M_1$ is continuous, the function
	$s \colon [0, 1] \rightarrow \mathbb{R}^2$ that describes the
	position of~$v_1$ in $\mathbb{R}^2$ relative to its starting point
	during the morph $M_1$ is continuous as well, and since~$v_1$ is
	drawn at the same position in~$\Gamma'_1$ and~$\Gamma'_2$, we know
	that $s(0) = s(1) = (0,0)$.  We now define a new morph~$M_2$ that is 
	obtained from $M_1$ by translating $M_1(x)$ by the vector $-s(x)$ 
	for every $x\in[0,1]$.
	Then $M_2$ is a continuous morph from~$\Gamma'_1$
	to~$\Gamma'_2$, and, for every $x \in [0,1]$, in the drawing $M_2(x)$, 
	the vertex $v_1$ is placed at $p_1$.
	
	In the same manner, we define the function
        $t \colon [0,1] \to \mathbb{R}_{>0}$ that describes the
        distance between $v_1$ and $v_2$ in $M_2(x)$ for every
        $x\in [0,1]$.  Since $d(p_1, p_2) = 1$, we know that
        $t(0) = t(1) = 1$, and $t$ is continuous since~$M_2$ is
        continuous.  We define a new morph~$M_3$ that is 
	obtained from $M_2$ by scaling $M_2(x)$ by the factor $1/t(x)$
        around the point~$v_1$ for every $x\in[0,1]$.
	Then $M_3$ is a continuous morph from~$\Gamma'_1$ to~$\Gamma'_2$
	that keeps~$v_1$ fix, and the distance between~$v_1$
        and $v_2$ is equal to 1 in $M_3(x)$ for every $x\in[0,1]$.
	
	We define a third function
        $\rho \colon [0,1] \rightarrow \mathbb{R}_{>0}$ that describes,
        for every $x\in [0,1]$, in~$M_2(x)$
        the angle between the horizontal ray that emanates from~$v_1$
        and points to the right and the ray from~$v_1$ to~$v_2$.
        Since $v_1$
        and $v_2$ are in the same positions in $\Gamma_1'$ and
        $\Gamma_2'$, we know that $\rho(0) = \rho(1) = 0$, and $\rho$
        is continuous since $M_3$ is continuous.  We define a new
        morph~$M_4$ that is obtained from $M_3$ by rotating $M_3(x)$
        around $v_1$ by an angle of $-\rho(x)$ for every $x\in[0,1]$.
        Then $M_4$ is a continuous morph from~$\Gamma'_1$
        to~$\Gamma'_2$, where $v_1$ is fix, the distance between~$v_1$
        and~$v_2$ equals~1 in $M_4(x)$ for every $x\in[0,1]$, and the
        the ray from~$v_1$ to~$v_2$
        does not change its slope during the morph.  Thus,
        $v_2$, too, is fix under~$M_4$.  Then, if we restrict $M_4$ to
        $G$, we obtain a morph from $\Gamma_1$ to $\Gamma_2$ that
        avoids $P$, which proves the statement for $|P| = 2$.
\end{proof}

In the remainder of this section, we focus on the case where the graph is a 3-cycle.
In particular, in the following two propositions, we show that a set of
compatible obstacles that blocks drawings of~$C_3$ has to contain
multiple internal {\em and} multiple external obstacles.  As we will
show in \cref{subsec:cycle-3}, the bounds on the number of internal
and external obstacles in these propositions are best-possible.	

\begin{proposition}
	\label{prop:block-label-shift-C3-one-outer}
	Let $\Gamma_1$ and $\Gamma_2$ be drawings of a plane $C_3$, and let
	$P$ be a set of obstacles compatible with~$\Gamma_1$ and~$\Gamma_2$.
	If at most one obstacle of~$P$ lies in the
	exterior of $\Gamma_1$ (and $\Gamma_2$), then there is a planar
	morph from $\Gamma_1$ to~$\Gamma_2$ that avoids~$P$.
\end{proposition}

\begin{proof}
Our strategy is as follows:
we first define a set of canonical drawings, then 
show how to transform any two canonical drawings into each other
and, finally, describe how each of the
given drawings can be transformed into a canonical drawing.
The combination of these three morphs yields the desired morph between
$\Gamma_1$ and $\Gamma_2$.

We describe our approach assuming that there is an external obstacle
(if this not the case, one can simply pick an arbitrary point in the exterior
of both drawings and, interpreting this point as an (external) obstacle, execute the
same strategy).

\paragraph{Canonical drawings.}
The following construction is illustrated in \cref{fig:canonical}.
Let~$o=(x_o,y_o)$ denote the external obstacle and its coordinates
and let $I$ denote the set of internal obstacles.
Without loss of generality, we may assume that $x_o$ is strictly smaller
than the x-coordinate of all internal obstacles.
(Note that $o$ lies in the exterior of the convex hull of $I$ and, hence, there exists
a line that separates~$o$ from $I$.
By applying a suitable rotation of the plane, we may assume this line to be
vertical and $o$ to be to its left.)
Let~$r_{\mathrm{above}}$ be a ray that originates at~$o$ and lies above
all internal
obstacles.
Symmetrically, let~$r_{\mathrm{below}}$ be a ray that originates at~$o$
and lies below all internal obstacles.
Note that the direction vector of both rays has a positive x-coordinate.
Let $a=(x_a,y_a)$, $b=(x_b,y_b)$, and $c=(x_c,y_c)$ be points (and their
coordinates) such that
\begin{itemize}
\item $x_a$ is strictly smaller
than the x-coordinate of all internal obstacles, but strictly larger than~$x_o$;
\item $x_b$ is strictly larger
than the x-coordinate of all internal obstacles;
\item $x_c=x_b$;
\item $a$ lies below~$r_{\mathrm{above}}$ and
above~$r_{\mathrm{below}}$;
\item $b$ lies above~$r_{\mathrm{above}}$ and such that the line through
$a$ and $b$ is above all obstacles; and
\item $c$ lies below~$r_{\mathrm{below}}$ and such that the line through
$a$ and $c$ is below all obstacles.
\end{itemize}
We use~$\ell_{ab}$ to denote the line through~$a$ and~$b$;
the lines~$\ell_{bc}$ and~$\ell_{ac}$ are defined analogously.
Note that all internal obstacles are located in the interior of the triangle~$abc$ while~$o$ is located in the exterior of~$abc$.
More precisely, the external obstacle~$o$ is located above~$\ell_{ab}$
and below~$\ell_{ac}$.

Any drawing of our plane 3-cycle in which the vertices are placed at $a$,
$b$, and~$c$ is called \emph{canonical}.
Thus, there are three distinct canonical drawings (taking into account shifted
versions).

\begin{figure}[tbp]
	\centering
	\includegraphics{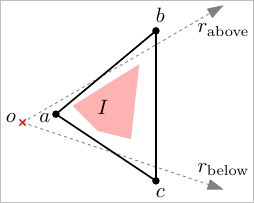}
	\caption{Construction of the canonical triangle.}
	\label{fig:canonical}
\end{figure}

\paragraph{Morphing between canonical drawings.}
Let~$\Gamma$ be a canonical drawing with vertices $v,w,u$ at $a,b,c$,
respectively.
We show how to morph to a shifted (canonical) drawing of~$\Gamma$.
(By repeating the described strategy once more, one can obtain a morph to the third
canonical drawing.)
The process is visualized in \cref{fig:rotation}.

Let~$p_v,p_u$ be points such that
\begin{itemize}
\item $p_v$ lies on~$\ell_{ac}$ and its x-coordinate is strictly smaller
than~$x_a$;
\item $p_u$ lies on~$\ell_{ab}$ and its x-coordinate is strictly smaller
than~$x_a$; and
\item the lines $\ell_{p_vb}$ and $\ell_{p_uc}$ through $p_v,b$ and $p_u,c$,
respectively, intersect in a point whose x-coordinate is strictly larger
than~$x_b(=x_c)$.
\end{itemize}
We use $p_w$ to denote the intersection point of $\ell_{p_vb}$ and $\ell_{p_uc}$.

\begin{figure}[tbp]
	\centering
	\includegraphics{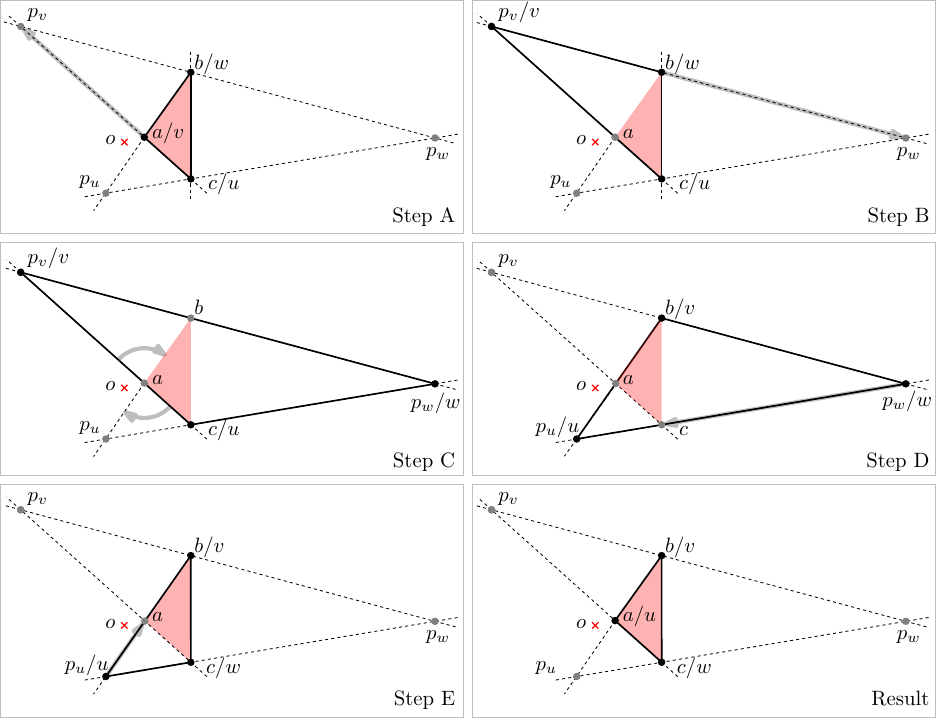}
	\caption{A morph of a canonical drawing of a 3-cycle with
          vertices $u$, $v$, $w$ placed at $c$, $a$, $b$,
          respectively, to a shifted version where the vertices are
          placed at $a$, $b$, $c$, respectively.  In the subfigures,
          we used the notation $x/y$ to mean that vertex~$y$ is placed
          at point~$x$.}
	\label{fig:rotation}
\end{figure}

\afterpage{\clearpage}

We proceed in five simple steps.
It is readily observed from properties of the triangle~$abc$ and the definitions of~$p_v,p_u,p_v$ that all of these
steps are obstacle-avoiding and planar.

	\vspace{-1.2ex}
	\begin{enumerate}[label={Step \Alph*.},leftmargin=*]
		\itemsep-0.2ex
		\item
		We use a linear morph to move~$v$ to~$p_v$.		
		\item
		We use a linear morph to move~$w$ to~$p_w$.	
		\item
		We simultaneously move $u$ to~$p_u$ along~$up_u$
		and~$v$ to~$b$ along~$p_vb$ and do so in a way where
		the intersection point of the edge~$uv$ with~$\ell_{ab}$ remains at~$a$.
		(Essentially, we rotate $uv$ around fixpoint~$a$ and simultaneously
		change its length to ensure that~$u$ and~$v$ move along the desired segments.
		We remark that this is not necessarily a linear morph.)
		\item
		We use a linear morph to move~$w$ to~$c$.
		\item
		We use a linear morph to move~$u$ to~$a$.
	\end{enumerate}
	\vspace{-1.2ex}

\paragraph{Morphing to canonical drawings.}
Let~$\Gamma\in \{\Gamma_1,\Gamma_2\}$.
Let~$\triangle$ be a triangle that has all internal vertices in its interior
and~$o$ in its exterior.
We describe an obstacle-avoiding planar morph that transforms~$\Gamma$
such that its vertices lie at the corners of~$\triangle$
(in particular, this means we can morph to a canonical drawing).
The process is visualized in \cref{fig:morph-to-canonical}.

\begin{figure}[pt]
	\centering
	\includegraphics{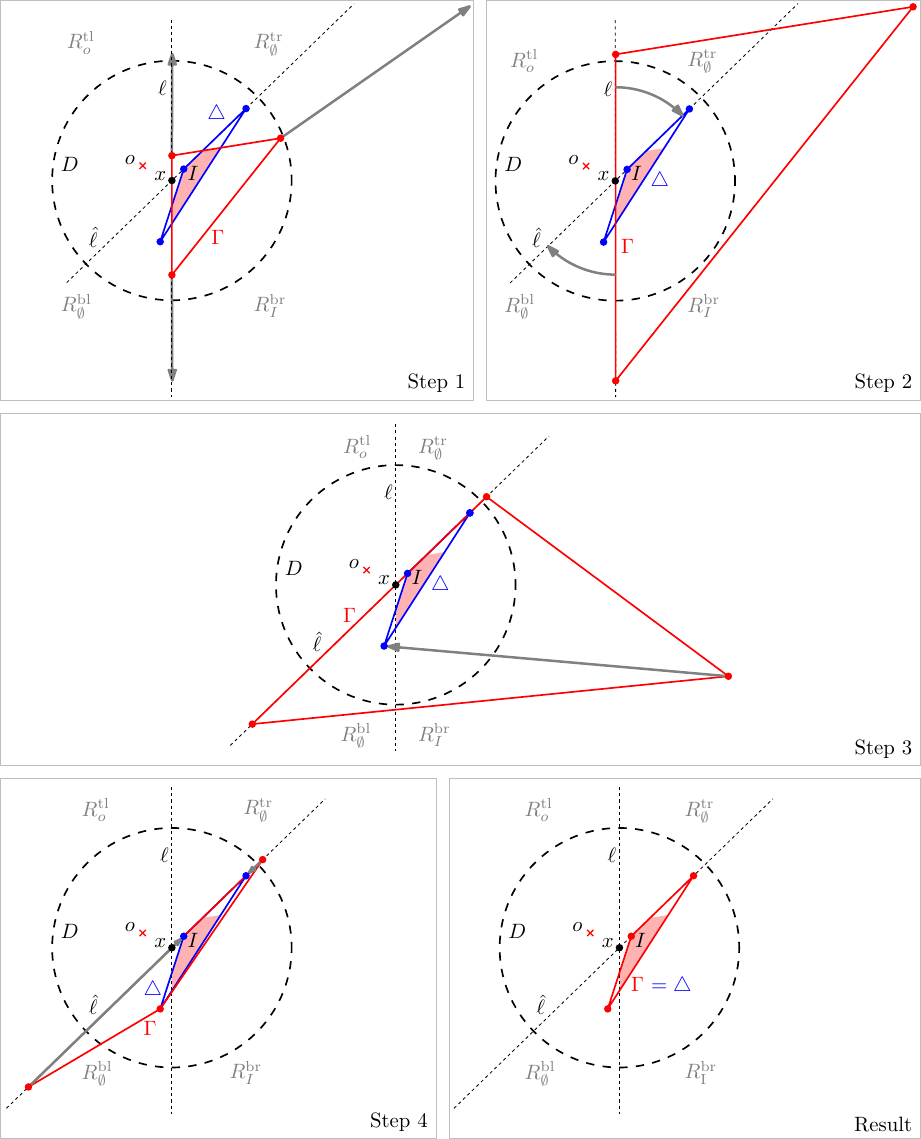}
	\caption{A morph transforming $\Gamma$ so that its vertices lie on the corners of a triangle $\triangle$. }
	\label{fig:morph-to-canonical}
\end{figure}

\afterpage{\clearpage}

Let~$e$ be an edge of~$\Gamma$ such that its supporting
line~$\ell$
has the property that~$o$ is on the side of~$\ell$ that
does not contain the internal obstacles.
Define an edge (side)~$\hat e$ of~$\triangle$ and its supporting line~$\hat \ell$
analogously.
We may assume that~$\ell$ and~$\hat \ell$ are not parallel, otherwise
we can simply use a morph to slightly move one of the endpoints of~$e$
to ensure this is the case (this can obviously be done in an obstacle-avoiding planar fashion).
Further, we may assume without loss of generality (by rotating and
/ or reflecting the plane if necessary) that
$\ell$ is vertical,
$o$ is to the left of~$\ell$, and
$o$ is above~$\hat \ell$.
By construction, the lines~$\ell$ and~$\hat \ell$ partition the plane into four regions,
which we denote as follows:
\begin{itemize}
\item
$R_o^{\mathrm{tl}}$ is to the left of~$\ell$, above~$\hat\ell$,
and contains~$o$;
\item
$R_\emptyset^{\mathrm{tr}}$ is to the right of~$\ell$, above~$\hat\ell$,
and contains no obstacles;
\item
$R_I^{\mathrm{br}}$ is to the right of~$\ell$, below~$\hat\ell$,
and contains all internal obstacles; and
\item
$R_\emptyset^{\mathrm{bl}}$ is to the left of~$\ell$, below~$\hat\ell$,
and contains no obstacles.
\end{itemize}
Let~$x$ denote the point where~$\ell$ and~$\hat\ell$ cross.
Let~$D$ be a disk centered on~$x$ that contains~$\Gamma$
and~$\triangle$ in its interior.%
Let~$u$ and~$v$ be the top and bottom endpoint of~$e$, respectively.
Let~$w$ denote the remaining vertex of~$\Gamma$
(which lies to the right of~$\ell$).
We proceed in several steps:

\vspace{-1.2ex}
\begin{enumerate}[label={Step \arabic*},leftmargin=*]
\itemsep-0.2ex
\item
Let~$\ell_{uw}'$ be a line that is parallel to the edge~$uw$ and that
lies above~$D$.
Symmetrically,
let~$\ell_{vw}'$ be a line that is parallel to the edge~$vw$ and that lies
below~$D$.
Let~$u'$ denote the point where~$\ell$ and~$\ell_{uw}'$ cross.
Let~$v'$ denote the point where~$\ell$ and~$\ell_{vw}'$ cross.
Finally, let~$w'$ denote the point where~$\ell_{uw}'$ and~$\ell_{vw}'$
cross.
We use a single linear morph to move~$u$ to~$u'$, $v$ to~$v'$,
and~$w$ to~$w'$.
Note that this morph ``expands'' our drawing, that is, the intial
drawing is contained in the interiors of the final drawing and all intermediate drawings.
Hence, it avoids the internal obstacles.
Moreover, the drawing remains in the half-plane to the right of~$\ell$
and, thus, the external obstacle~$o$ is also avoided.
\item
Note that the edges~$uw$ and~$vw$ are now in the exterior of~$D$.
We use a morph
to rotate our drawing clockwise around fixpoint~$x$ until the edge~$uv$
lies on~$\hat\ell$.
During this rotation, the edges~$uw$ and~$vw$ remain in the exterior
of~$D$ and, thus, avoid all internal obstacles (which all lie in the
interior of~$D$).
Moreover, the edges~$uw$ and~$vw$ never enter the region~$R_o^{\mathrm{tl}}$ and, thus avoid the external obstacle~$o$.
The supporting line of the edge~$uv$ covers exactly the regions
$R_\emptyset^{\mathrm{tr}}$ and $R_\emptyset^{\mathrm{bl}}$
during the rotation.
Thus, the edge~$uv$ also avoids all obstacles.
\item
We use a linear morph to move~$w$ to the corner of~$\triangle$
that does not lie on~$\hat\ell$.
It is easy to see that this morph is obstacle-avoiding and planar.
\item
We use a linear morph to move~$u$ to the right corner of~$\triangle$
on~$\hat\ell$ and~$v$ to the left corner of~$\triangle$ on~$\hat\ell$.
It is easy to see that this morph is obstacle-avoiding and planar.
\end{enumerate}
\vspace{-1.2ex}

		\paragraph{Wrap-up.}
	We have described how to transform both $\Gamma_1$ and $\Gamma_2$ into canonical drawings and how canonical drawings can be transformed into each other.
	In combination, this yields the desired morph.
\end{proof}

\begin{proposition}
	\label{prop:block-label-shift-C3-two-inner}
	Let $\Gamma_1$ and $\Gamma_2$ be drawings of a plane~$C_3$,
        and let~$P$ be a set of obstacles compatible with~$\Gamma_1$
        and~$\Gamma_2$.
	If at most two obstacles of~$P$ lie in the interior of
	$\Gamma_1$ (and $\Gamma_2$), then there is a planar
	morph from $\Gamma_1$ to~$\Gamma_2$ that avoids~$P$.
\end{proposition}

\begin{proof}
	Our strategy is similiar to the one used in \cref{prop:block-label-shift-C3-one-outer}:
	we construct a rectangle that contains both internal obstacles and that
	is in the interior of both drawings.
	We then describe how each of our drawings can be transformed such that its vertices lie on the boundary of the rectangle.
	Finally, we describe how to morph between the resulting two drawings.
	The combination of these three morphs yields the desired morph
        between~$\Gamma_1$ and~$\Gamma_2$.
        
        We describe our approach assuming that there are exactly two internal
        obstacles.
        If there is only one internal obstacle, we can pick an arbitrary point interior to
        both drawings and, interpreting this point as an additional (internal) obstacle,
        execute
        the same strategy.
       Similarly, if there is no internal obstacle, we can simply shrink the drawing~$\Gamma_1$
	(by moving all of its vertices close to some point in its interior),
	use a sequence of translations to move it
	into the interior of~$\Gamma_2$ without intersecting any obstacles,
	pick two arbitrary points in the interior of both drawings, and, interpreting these points as (internal) obstacles, execute our strategy for two internal obstacles.

	We begin by introducing some notation.
	
	\paragraph{Setup.}
	For illustrations of the notation introduced in this paragraph, refer to \cref{fig:triangle-4-obstacles-2}.
	We let $p_a=(x_a,y_a)$ and $p_b=(x_b,y_b)$ denote the internal
        obstacles and their coordinates.
	Without loss of generality, we assume that
 $x_a=x_b$ and $y_a>y_b$.
	Choose an~$\varepsilon >0$ such that the rectangle
	$R = [x_a-\varepsilon,x_a+\varepsilon] \times
	[y_b-\varepsilon,y_a+\varepsilon]$ is contained in the interior of both~$\Gamma_1$ and~$\Gamma_2$.
	Let $\ell$, $m$, and $r$ denote the vertical lines with
        x-coordinates $x_\ell=x_a-\varepsilon$, $x_m=x_a$, and
        $x_r=x_a+\varepsilon$, respectively.
	Let $\Gamma \in \{\Gamma_1, \Gamma_2\}$, and let $u=(x_u,y_u)$,
	$v=(x_v,y_v)$, and $w=(x_w,y_w)$ denote (the coordinates of) its vertices.
	Note that~$\Gamma$ intersects~$m$ in exactly two points.
	In other words, one of its three edges, say~$uv$, lies completely in one of the two closed half-planes bounded by~$m$ and the other two edges each intersect~$m$ in a single point (possibly an endpoint).
	By suitably reflecting the plane and relabeling obstacles and / or vertices if necessary, we may assume that
	\begin{itemize}
	\item if one endpoint of $uv$ lies on~$m$, then this endpoint is~$v$;
\item $uv$ lies in the closed half-plane to the left of~$m$, i.e., $x_u< x_m$, $x_v\le x_m$ (and, consequently $x_r<x_w$); and
\item the edge $vw$ intersects~$m$ above~$R$ and the edge $uw$ intersects~$m$ below~$R$.
\end{itemize}

	\begin{figure}[!tbh]
		\centering
		\includegraphics[page=2]{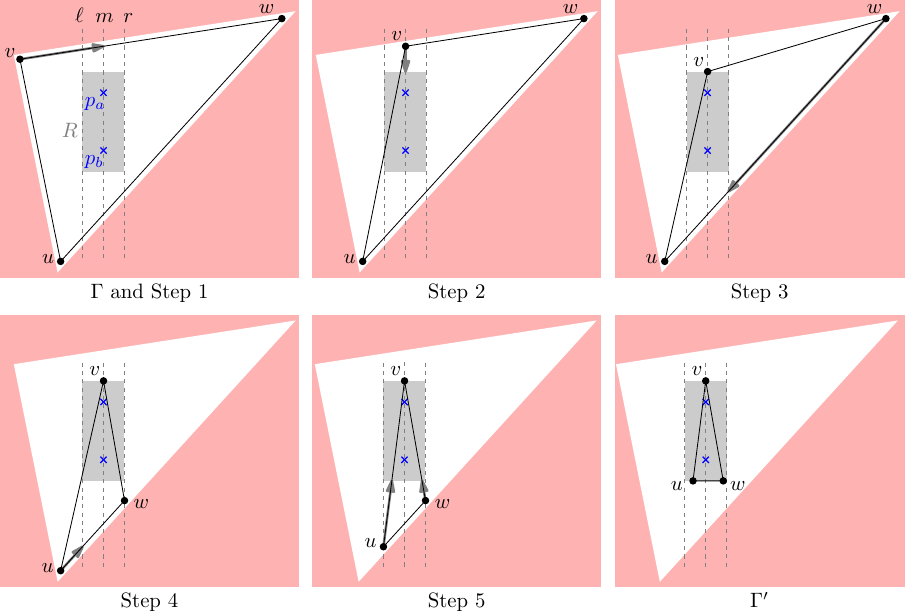}
		\caption{
		Our strategy for morphing the vertices of~$\Gamma$
		onto the boundary of the rectangle~$R$ in the proof of
		\cref{prop:block-label-shift-C3-two-inner}.
		The resulting morph avoids the exterior of the initial drawing
		(shaded) entirely
		and can hence be applied regardless of the number and
		placement of the external obstacles.}
		\label{fig:triangle-4-obstacles-2}
	\end{figure}

\paragraph{Moving the vertices onto the boundary of the rectangle.}
We now describe an obstacle-avoiding planar morph that transforms
$\Gamma$ into a drawing whose vertices lie on the boundary of~$R$.
It consists of several basic steps, which are illustrated in
\cref{fig:triangle-4-obstacles-2}:

	\vspace{-1.2ex}
	\begin{enumerate}[label={Step \arabic*.},leftmargin=*]
		\itemsep-0.2ex
		\item
		Move $v$ along $vw$ to the intersection of $vw$ with $m$ (which is above~$R$).
		Note that the (triangular) convex hull of the old and new drawing of $uv$ (this is the region covered by $uv$ during this morph) does not contain any obstacle or $w$, so this morph is planar and obstacle-avoiding.
\item Move $v$ downwards along $m$ to the top boundary of~$R$.
Note that the (triangular) convex hull of the old and new drawing of $uv$ does not contain any obstacle or $w$ and the (triangular) convex hull of the old and new drawing of $vw$ does not contain any obstacle or $u$, so this morph is planar and obstacle-avoiding.
\item Move $w$ along $uw$ to the intersection of $uw$ with $r$ (which is below $R$).
Note that the (triangular) convex hull of the old and new drawing of $vw$ does not contain any obstacle or $u$, so this morph is planar and obstacle-avoiding.
\item Move $u$ along $uw$ until it lies in the closed vertical slab bounded by~$\ell$ and~$m$; it then lies below $R$ (we remark that it is possible for $u$ to lie in this slab to begin with).
Note that the (triangular) convex hull of the old and new drawing of $uv$ does not contain any obstacle or $w$, so this morph is planar and obstacle-avoiding.
\item Finally, move $u$ along $uv$ and $w$ along $vw$ until they reach the bottom segment of~$R$.
Note that the (quadrangular) convex hull of the old and new drawing of $uw$ does not contain any obstacle or $v$, so this morph is planar and obstacle-avoiding.
	\end{enumerate}
	\vspace{-1.2ex}
	
	Let~$\Gamma'$ denote the result drawing.
	
	\paragraph{Morphing to and between canonical drawings.}
	We say that a drawing of our plane 3-cycle is \emph{canonical} if its vertices
	are placed either at
	$R_{\mathrm{bl}}=(x_b-\varepsilon,y_b-\varepsilon)$,
	$R_{\mathrm{br}}=(x_b+\varepsilon,y_b-\varepsilon)$, and
	$R_{\mathrm{tm}}=(x_a,y_a+\varepsilon)$, or, symmetrically, at
	$R_{\mathrm{tl}}=(x_a-\varepsilon,y_a+\varepsilon)$,
	$R_{\mathrm{tr}}=(x_a+\varepsilon,y_a+\varepsilon)$, and
	$R_{\mathrm{bm}}=(x_b,y_b-\varepsilon)$.
	Taking into account shifted drawings, there are thus $6$ canonical drawings in total.
	
	It is easy to transform our drawing~$\Gamma'$ into a canonical one:
	vertex~$v$ is already located at $R_{\mathrm{tm}}$ and
	the other two vertices are located on the bottom segment of~$R$ (with~$u$ to the left of~$w$).
	Thus, we can simply move $u$ to $R_{\mathrm{bl}}$ and $w$ to $R_{\mathrm{br}}$ along the bottom segment of~$R$.
	Moreover, we can now easily transform the resulting canonical drawing~$\Gamma''$ into any other canonical drawing:
	first, we slightly increase the x-coordinate of~$v$ while maintaining the property that the line through $u$ and $v$ is above the obstacles $p_a$ and $p_b$.
	Second, we move~$w$ to $R_{\mathrm{bm}}$.
	Third, we move~$u$ to $R_{\mathrm{tl}}$.
	Finally, we move~$v$ to $R_{\mathrm{tr}}$.
	In each step, the (triangular) convex hull of the old and new version of redrawn edge does not contain any obstacles or the nonincident vertex, so all these morphing steps are planar and obstacle-avoiding.
	By repeating (a symmetric version of) this strategy, we can reach any canonical drawing, as claimed.
	
	\paragraph{Wrap-up.}
	We have described how to transform both $\Gamma_1$ and $\Gamma_2$ into canonical drawings and how canonical drawings can be transformed into each other.
	In combination, this yields the desired morph.
\end{proof}

\afterpage{\clearpage}

\section{A Tight Lower Bound on the Number of Obstacles to Block Drawings of 3-Cycles}
\label{subsec:cycle-3}

In this section, we establish a tight lower bound on the number
 of obstacles needed
to block (shifted) drawings of~$C_3$.

\begin{proposition}
	\label{prop:block-label-shift-C3}
	Let $\Gamma_1$ and $\Gamma_2$ be drawings of a plane $C_3$, and let
	$P$ be a set of obstacles compatible with~$\Gamma_1$ and~$\Gamma_2$.
	If $\Gamma_1$ and $\Gamma_2$ are blocked by~$P$,
	then~$P$ contains at least five obstacles
	(three of which are internal and two of which are external).
\end{proposition}

\begin{proof}
Follows immediately from \cref{prop:block-label-shift-C3-two-inner,prop:block-label-shift-C3-one-outer}.
\end{proof}

\begin{proposition}
	\label{prop:block-label-shift-C3-2}
	Let~$\Gamma$ be a drawing of~$C_3$, and let~$\Gamma'$ be a shifted
	version of~$\Gamma$.  Then $\Gamma$ and $\Gamma'$ are blockable by
	five obstacles (three of which are internal and two of which
    are external) compatible with~$\Gamma$ and~$\Gamma'$.
    
    If $\Gamma$ is an isosceles triangle where the base has length~$s$, then it suffices to place
    the internal obstacles at a distance of at most~$s/12$ next to the three vertices and
    the external obstacles at a distance of at most~$s/12$ next to the midpoints of the two legs.
\end{proposition}

\begin{proof}
        Let $a$, $b$, and $c$ be the vertices of~$C_3$.  We put three
        obstacles $a'$, $b'$, and $c'$ in the inner face of~$\Gamma$
        in the vicinity of the corresponding vertices, i.e., inside a
        disk of radius $\varepsilon > 0$, where $\varepsilon$ is much
        smaller than the length of any edge of~$\Gamma$; see
        \cref{fig:triangle-label-shift}.  We put the other two
        obstacles~$d'$ and~$e'$ in the outer face of~$\Gamma$, namely
        in the $\varepsilon$-vicinity of the midpoints of the edges~$ac$ and~$ab$,
        respectively.  We show by contradiction that this set of
        obstacles blocks any morph between~$\Gamma$ and~$\Gamma'$.
	
	\begin{figure}[tb]
		\centering
		\includegraphics[page=5]{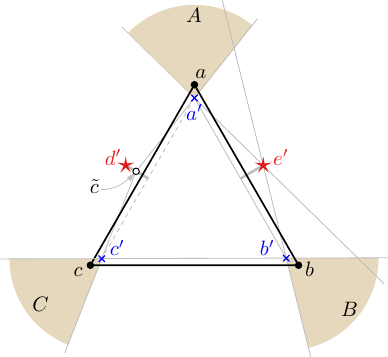}
		\caption{Five obstacles suffice to block shifted
			versions of~$C_3$.}
		\label{fig:triangle-label-shift}
	\end{figure}
	
	Without loss of generality, we may assume that $a$ is the
        vertex of~$C_3$ with the largest y-coordinate in~$\Gamma$,
        that~$b$ lies on the right side of the vertical line
        through~$a$, and that~$c$ lies on the left side.  We may also
        assume that $\varepsilon$ is chosen so small that the line
        $c'd'$ has positive slope and the line $b'e'$ has negative
        slope.  This implies that~$b'$ lies on the right side and that
        $c'$ lies on the left side of the vertical line through~$a$.
        
        Now suppose, for a contradiction, that there is a morph
        from~$\Gamma$ to a shifted version $\Gamma'$ such that $c$ is
        the vertex with largest y-coordinate in~$\Gamma'$.
        Note that during this morph no edge of the triangle
        $\triangle abc$ may intersect the triangle $\triangle a'b'c'$
        because $\triangle abc$ must always contain~$a'$, $b'$,
        and~$c'$.  Hence, on the way from its initial to its
        final position, vertex~$c$ must go around $\triangle a'b'c'$
        either clockwise (near~$d'$) or counterclockwise (near~$e'$).
        We assume that $c$ goes clockwise; the other case is
        symmetric.

        Consider the intermediate drawing~$\tilde\Gamma$ where $c$
        lies on the normal of the edge~$ac$ through~$d'$ in~$\Gamma$.
        Let~$\tilde a$, $\tilde b$, and $\tilde c$ be the positions
        of~$a$, $b$, and~$c$ in $\tilde\Gamma$, respectively.  On the
        normal of~$ac$, $\tilde c$ must lie between~$d'$ and the
        intersection point with~$a'c'$.
        For the edge~$\tilde c \tilde a$ not to
        intersect $\triangle a'b'c'$, $\tilde a$ must lie above the
        line~$\tilde ca'$; see \cref{fig:triangle-label-shift}.  To
        ensure that $b'$ lies inside
        $\triangle \tilde a \tilde b \tilde c$, the point~$\tilde b$
        must be below the line $b'c'$.  Similarly, the point
        $\tilde a$ must lie above the line $a'e'$ to ensure that~$e'$
        lies outside $\triangle \tilde a \tilde b \tilde c$.  Hence
        $\tilde a$ must lie inside the wedge~$A$ with apex~$a'$ that
        is the intersection of the two halfplanes above the
        lines~$\tilde ca'$ and~$e'a'$.  Similarly, from the point
        of view of~$\tilde c$, the point~$\tilde b$ must lie in the
        wedge~$C$ with apex~$c'$ that is the intersection of the two
        halfplanes below the line~$b'c'$ and the halfplane to the left
        of the line~$c'\tilde c$.  To ensure that~$e'$ lies outside
        $\triangle \tilde a \tilde b \tilde c$, however, the
        point~$\tilde b$ must lie to the right of the line~$b'e'$.
        This forces $\tilde b$ to lie in the
        wedge~$B$ with apex~$b'$, that is, the intersection of the
        halfplane below~$b'c'$ and the halfplane to the right
        of~$e'b'$.  Hence $\tilde b$ must lie in $B \cap C$.  However,
        the apices~$b'$ and~$c'$ of~$B$ and~$C$, respectively, lie on
        different sides of the vertical line through~$a$.  Moreover,
        the left boundary of~$B$ is part of the line~$b'e'$, whose
        slope we assumed to be negative, and the slope of the right
        boundary of~$C$ is less than the slope of the line~$c'd'$,
        which we assumed to be positive.  Therefore, $B$ lies
        completely to the right of the vertical line through~$a$,
        whereas~$C$ lies completely to the left of this line.

	Hence, the intermediate drawing $\tilde\Gamma$ cannot exist,
	and $\Gamma$ and $\Gamma'$ are	blocked by the set of obstacles.
    
    \begin{figure}[t]
        \centering
        \includegraphics[page=3]{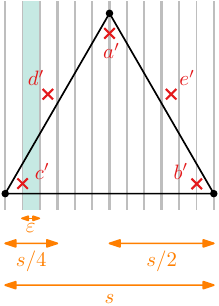}
        \caption{If we place the obstacles as described in
            the proof of \cref{prop:block-label-shift-C3-2}
            in $\varepsilon$ distance around an isosceles triangle
            with side lengths~$s$,
            then it suffices to set $\varepsilon = s/12$ to ensure
            that the line segment~$c'd'$ has positive slope.}
        \label{fig:epsilon-is-large}
    \end{figure}
    
    Finally, we argue about the precise placement of the obstacles stated
    above if $\Gamma$ is an isosceles triangle whose base has length~$s$.
    To this end, we show how to choose $\varepsilon$ sufficiently small.
    Without loss of generality, we assume that $bc$ is the base and
    $bc$ is drawn horizontally in $\Gamma$.
    The only requirement to $\varepsilon$ is that the lines $c'd'$ and $b'e'$
    have a positive and a negative slope, respectively.
    If we set $\varepsilon = s/12$, it is easy to see that these
    requirements are fulfilled:
    if we divide the drawing into vertical strips of width $s/12$, then
    $c'$ is in a strip to the left of~$d'$, and
    $e'$ is in a strip to the left of $b'$; see \cref{fig:epsilon-is-large}.
\end{proof}

\section{Shifted Versions of Cycles with and without Free Vertices}
\label{subsec:cycle-n}

In this section, we investigate (shifted) drawings cycles of arbitrary length.
We start by showing that it is not always possible to morph between a drawing
of $C_n$ and a shifted version of this drawing in an obstacle-avoiding fashion.

\begin{proposition}
	\label{prop:block-label-shift-in-even-cycles}
	Let $n \ge 6$ be an even integer.  Then there exists a
	drawing~$\Gamma$ of~$C_n$ such that, for every shifted
	version~$\Gamma'$ of~$\Gamma$, the drawings $\Gamma$ and $\Gamma'$ are
	blockable by seven obstacles that are compatible with~$\Gamma$
	and~$\Gamma'$.
\end{proposition}

\begin{proof}
	Let $\Gamma$ be a drawing of $C_n$ that it is wrapped around an
	imaginary equilateral triangle~$\triangle$ formed by the positions
	of the first three vertices of~$C_n$; see 
	\cref{fig:blocked-cycle-7-obstacles}.  We place~$\triangle$
	such that its bottom edge is horizontal.  We start with the
	first vertex $v_1$ of~$C_n$ in the lower right corner of~$\triangle$.
	Then, in $n/2$ steps, we go clockwise around (the corners of)~$\triangle$.
	In the $i$-th step, we place~$v_{i+1}$ and $v_{n-i+1}$ near the currently considered corner
	of~$\triangle$ such that the two new vertices see their predecessors
	(and will see their successors) and lie slightly further away from
	the barycenter of~$\triangle$ than the other vertices that we have
	placed near the same corner of~$\triangle$.  We can place the
	vertices of~$C_n$ (and four of the obstacles, see below)
	such that each of them lies within a
	radius-$\varepsilon$ disk centered at one of the corners
	of~$\triangle$, for some $\varepsilon>0$.  These
	radius-$\varepsilon$ disks correspond to the shaded green areas in the
	schematic drawing in \cref{fig:blocked-cycle-7-obstacles}.
	
	\begin{figure}[tbh]
		\centering
		\includegraphics{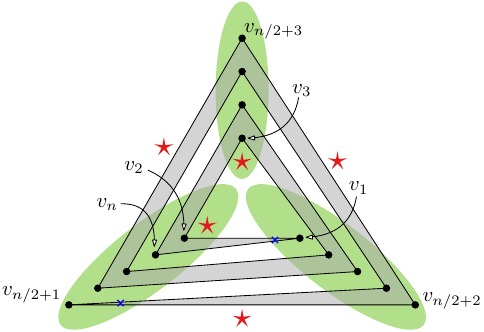}
		\caption{A schematic drawing of $C_n$ with vertices
			$v_1,\dots,v_n$ (here $n=14$) that cannot be morphed to a
			shifted version of itself in a planar
			way while avoiding the two internal obstacles (blue crosses) and
			the five external obstacles (red stars).
            In the actual drawing, each set of vertices and obstacles
            contained in the same shaded green area is placed within a
            radius-$\varepsilon$ disk (for some $\varepsilon>0$) centered
            at a corner of an equilateral triangle with unit edge length.
		}
		\label{fig:blocked-cycle-7-obstacles}
	\end{figure}
	
	We place two obstacles in the interior of~$C_n$, namely in the
	immediate vicinity of~$v_1$ and~$v_{n/2+1}$ (see the blue crosses in
	\cref{fig:blocked-cycle-7-obstacles}).  Additionally, we place five
	obstacles in the complement of~$C_n$ (see the red stars in
	\cref{fig:blocked-cycle-7-obstacles}).  First, we place two of these
	in the immediate vicinity of~$v_2$ and~$v_3$
	(closer to $v_2$ and $v_3$ than to $v_n$ and $v_{n-1}$, respectively).
	Then, we place the
	remaining three obstacles very close to the midpoints of the three
	outermost edges of~$C_n$ (that is, $v_{n/2+1}v_{n/2+2}$,
	$v_{n/2+2}v_{n/2+3}$, and $v_{n/2+3}v_{n/2+4}$).
	Due to the tight wrapping of~$C_n$ around~$\triangle$ and the extra
	obstacles on each of the three sides of the drawing,
	the drawing of $C_n$ is constrained within a triangular strip 
	of width~$\varepsilon$ around $\triangle$.
	Therefore no edge can change its slope (measured as the angle
        it forms with, say, the positive x-axis) by more than
        $O(\varepsilon)$ degrees.  Hence no edge can be turned such that it
        aligns with one of its incident edges, which would mean to
        change its slope by roughly $120^{\circ}$.
	Hence, $\Gamma$ and $\Gamma'$ are blocked 
	for every shifted version $\Gamma'$ of $\Gamma$.
\end{proof}

Observe that it is plausible that
\cref{prop:block-label-shift-in-even-cycles} can be strengthened:
even three obstacles (the two internal and one external in the
bottom left corner of the middle part of the drawing; see
\cref{fig:blocked-cycle-7-obstacles}) seem to be sufficient
for blocking two shifted drawings of an even-length cycle (we chose
to use seven obstacles to simplify the proof).
In view of this remark, it seems that
\cref{prop:1or2-obstacles} is best-possible.
In particular, the proof strategy does not generalize to three
obstacles as affine transformations preserve the cyclic
orientations of point triples.

Now we show that the presence of a free vertex is sufficient for the
existence of planar obstacle-avoiding morphs between shifted
drawings of a cycle.

\begin{proposition}
	\label{prop:free-vertex-label-shift}
	Let~$\Gamma$ be a drawing of~$C_n$, and let~$\Gamma'$ be a shifted
	version of~$\Gamma$.  If~$\Gamma$ contains a free vertex, then
	$\Gamma$ and~$\Gamma'$ are not blockable by obstacles that are
	compatible with~$\Gamma$ and~$\Gamma'$.
\end{proposition}

\begin{proof}
	Let $v_0,v_1\dots,v_{n-1}$ be the vertices of $C_n$ in clockwise
	order around $\Gamma$ such that $v_0$ is the free vertex.  We
	move~$v_0$ immediately next to~$v_1$
	such that $v_1$ lies on the line segment~$v_0v_2$.
    As a result, $v_1$ is free and can be moved very close
    to~$v_2$.  We iterate this move for every vertex until $v_0$
    becomes free again.  Finally, we move, for $i \in [n-1]$,
    $v_i$ very slightly to the former position of~$v_{i+1}$
    and~$v_0$ to, say, the midpoint of the line segment~$v_{n-1}v_1$.
    Together, all these moves represent
    a morph from~$\Gamma$ to a shifted version~$\Gamma_1$ of
	$\Gamma$, where, for every $i \in [n-2]$, vertex $v_i$ has moved to
	the position of its successor~$v_{i+1}$ and vertex $v_{n-1}$ has
	moved to the position of~$v_1$.  Repeating the whole procedure as
	often as necessary yields a morph from $\Gamma$ to any shifted
	version of $\Gamma$.
\end{proof}

\section{NP-Hardness (Proof of \cref{thm:NP-hard*})}
\label{sec:hard}

In this section, we show our main result.

\Hardness*
\label{thm:NP-hard}

\begin{proof}
	We reduce from the classical NP-hard problem \textsc{3-SAT}.
	Given a Boolean formula $\Phi = \bigwedge_{i=1}^m c_i$ in conjunctive
        normal form over variables $x_1, x_2, \dots, x_n$ whose clauses
	$c_1, c_2, \dots, c_m$ consist of three literals each, we construct
	a plane graph~$G$, two planar
	drawings~$\Gamma$ and~$\Gamma'$ of~$G$, and a set~$P$ of obstacles
	that are compatible with~$\Gamma$ and~$\Gamma'$ such that there
	exists an obstacle-avoiding planar morph from $\Gamma$ to $\Gamma'$
	if and only if $\Phi$ is satisfiable.
	We proceed with a high-level overview of our construction and then discuss the individual aspects in more detail.
	
	\begin{figure}[h]
		\centering
		\includegraphics[]{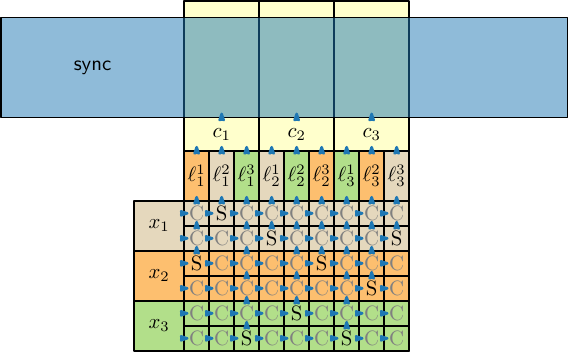}
		\caption{General grid structure used in our NP-hardness reduction.
			Here, we use the formula $\Phi = c_1 \land c_2 \land c_3$,
			where $c_1 = (\ell_1^1 \lor \ell_1^2 \lor \ell_1^3) = (x_2 \lor x_1\lor \neg x_3)$,
			$c_2 = (\ell_2^1 \lor \ell_2^2 \lor \ell_2^3) = (\neg x_1 \lor x_3 \lor x_2)$, and
			$c_3 = (\ell_3^1 \lor \ell_3^2 \lor \ell_3^3) = (\neg x_3 \lor \neg x_2 \lor \neg x_1)$.
			There are variable gadgets (left), clause and literal gadgets (top),
			split gadgets (S), crossing gadgets~(C),
			and a synchronization gadget (\textsf{sync})
			spanning over all clause gadgets (blue rectangle).
			The gadgets have various states and orientations.
			The blue arrows mark in which direction the truth state of a satisfied literal {\it can} be propagated.
			The colors in the rows, i.e., brown, orange, green,
			correspond to the three variables $x_1$, $x_2$, and~$x_3$,
			respectively.  The literal gadgets above have the same colors
			as the corresponding variables.}
		\label{fig:hardness-sketch}
	\end{figure}
	
	\paragraph{Overview.}
	We arrange the obstacles of~$P$ and the vertices and edges in~$\Gamma$
	and~$\Gamma'$ such that we obtain a grid-like structure where
	we have two rows for each variable (one for each literal of the variable)
	and three columns for every clause (one for each literal in the clause);
	see \cref{fig:hardness-sketch}.
	We then use several gadgets arranged within this grid-like structure.
	On the left side, the two rows of a variable
	terminate at a \emph{variable gadget}.
	A variable gadget is in one of the states
	\textsf{true}, \textsf{false}, or \textsf{unset}.
	On the top side, the three columns of a clause
	are joined via three \emph{literal gadgets}
	to a \emph{clause gadget}.
	Each literal gadget and each clause gadget
	is in one of the states \textsf{true} or \textsf{false}.
	All clause gadgets are joined by a \emph{synchronization gadget}.
	Within the column of each literal,
	we have a \emph{split gadget} in one of the two rows
	of the corresponding variable~$x_i$---in the upper row
	if the literal is $x_i$ and
	in the lower row if the literal is $\lnot x_i$.
	In all other grid cells, we have \emph{crossing gadgets}.
	Each of the split and crossing gadgets has four possible
	states,
	which we express in terms of both
	a horizontal orientation (\textsf{left}/\textsf{right})
	{\it and} a vertical orientation (\textsf{bottom}/\textsf{top}), i.e.,
	such a gadget can be oriented right and top, right and bottom,
	left and top, or left and bottom.
	
	The general idea behind the construction is as follows:
	the initial drawing $\Gamma$ and the final drawing
	$\Gamma'$ differ only in the synchronization gadget; the rest of
	the construction is identical.
	Transforming the synchronization gadget from its initial to
	its final state is possible if and only if all clause gadgets are in
	their \textsf{true} state at the same time -- both in $\Gamma$
	and $\Gamma'$ all clause gadgets are however in their
	\textsf{false} state.
	A clause gadget can be in (or transformed into) its
	\textsf{true} state if and only if at least one of its three literal
	gadgets is in its \textsf{true} state -- just like the clause
	gadgets, both in $\Gamma$ and $\Gamma'$ all literal gadgets
	are in their \textsf{false} state.
	A literal gadget can be in (or transformed into) its
	\textsf{true} state if and only if the state of the gadget of the
	variable of the literal is set accordingly (to \textsf{true} in case
	of positive literals and to \textsf{false} in case of negative
	literals).
	Both in $\Gamma$ and $\Gamma'$, the state of all variable
	gadgets is \textsf{unset}, but they can be morphed easily (and independent of each other) to
	achieve any of their other two states.
	However, since a variable gadget is not directly joined with all of its
	literal gadgets, the state of the variable gadget needs to be propagated
	to the correct literal gadgets.
	To this end, we use crossing and split gadgets, which
	we discuss in the following paragraph.
	
	A crossing or split gadget $\alpha$ can have the
	\textsf{right} horizontal orientation if and only if (1) the literal
	corresponding to its row is satisfied and (2) 
	all gadgets strictly between $\alpha$ and the variable
	gadget in this row also
	have the \textsf{right} horizontal orientation --
	intuitively, the information that the literal is satisfied is
	transported from the variable gadget in a rightwards direction to
	$\alpha$ (and all gadgets in between).
	If these conditions are satisfied and $\alpha$ is a split gadget,
	then $\alpha$ can also have the \textsf{top} vertical orientation.
	The vertical orientation works similar to the horizontal orientation;
	a crossing or split gadget $\alpha$
	that is located above the split gadget~$s$ of its column can
	have the \textsf{top} vertical orientation if and only if
	(1) $s$ has the \textsf{top} vertical orientation, and
	(2) all gadgets in the column strictly between $\alpha$ and $s$
	have the \textsf{top} vertical orientation.
	Intuitively, the split gadget needs to receive
	the information that the literal in its row is satisfied via its
	horizontal orientation in order to propagate this information in
	the upwards direction to $\alpha$ (and all gadgets in between).
	In particular, a split or crossing gadget that is directly below
	a literal gadget can have the vertical orientation \textsf{top} if
	and only if the literal corresponding to the literal gadget
	is satisfied.
	
	In summary, we can morph from $\Gamma$ to $\Gamma'$
	if and only if we can transform the synchronization gadget,
	which is possible if and only if all clause gadgets are in the
	\textsf{true} state at the same time, which is possible if and
	only if at least one literal gadget of each clause gadget is in the
	\textsf{true} state at the same time, which due to the
	propagation mechanism is possible if and only if the variables
	gadgets represent a satisfying truth assignment at the same
	time.
	(Notably, once the synchronization gadget has been transformed, it can remain in its altered state even if the other gadgets are morphed back to their states in $\Gamma$ and $\Gamma'$).
	
	We proceed by discussing the low-level details of our
	construction.

	\paragraph{Forbidden areas.}
	In each gadget, we have \emph{forbidden} areas where the vertices and
	edges cannot be drawn (henceforth drawn solid red).
	They are used to  create a system of tunnels and cavities
	in which the edges of our drawings are placed and move; see, e.g.,
	\cref{fig:nph-full-construction}. %
	We achieve this by populating the forbidden areas with obstacles, 
	arranged on a fine grid (as explained in more detail
	in the paragraph ``Running time''
	on page \pageref{par:number-obstacles}).
	
	\paragraph{Variable gadget.}
	The variable gadget has a comparatively simple structure;
	see \cref{fig:variable-gadget}.
	It has three vertices enclosed in a straight vertical tunnel of the forbidden area
	with one exit on the top right and one exit on the bottom right.
	Since the tunnel is straight, the middle vertex~$d$,
	which we call \emph{decision vertex}, is free (see
	\cref{sec:intro} for the definition of free and see the thick
	blue--white vertex~$x_3$ in \cref{fig:nph-full-construction} for
	the position of a free vertex in~$\Gamma$ and~$\Gamma'$).
	We say that the variable gadget is in the state \textsf{true} (\textsf{false})
	if $d$ is moved to the top (bottom) position,
	and it is in the state \textsf{unset} otherwise.
	Consequently, we can move the top (bottom) vertex out of the vertical tunnel
	if and only if we are in the state \textsf{true} (\textsf{false}).
	We call a vertex that can be moved from its original position within a gadget
	into an adjacent gadget a \emph{transferable} vertex.
	In the case of a variable gadget, one of the (black) neighboring vertices of~$d$
	becomes a transferable vertex if the variable gadget is in
	the state \textsf{true} or \textsf{false}
	(see \cref{fig:variable-gadget-true,fig:variable-gadget-false}).

    \begin{figure}[tbh]
        \centering
        \begin{subfigure}[t]{.3\linewidth}
            \centering
            \includegraphics[page=19]{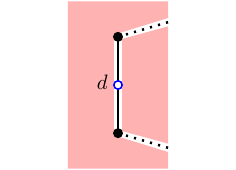}
            \caption{\textsf{unset}}
            \label{fig:variable-gadget-unset}
        \end{subfigure}
        \hfill
        \begin{subfigure}[t]{.3\linewidth}
            \centering
            \includegraphics[page=20]{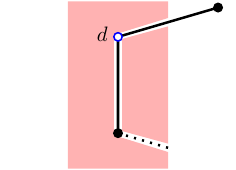}
            \caption{\textsf{true}}
            \label{fig:variable-gadget-true}
        \end{subfigure}
        \hfill
        \begin{subfigure}[t]{.3\linewidth}
            \centering
            \includegraphics[page=21]{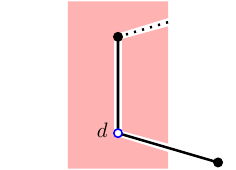}
            \caption{\textsf{false}}
            \label{fig:variable-gadget-false}
            \end{subfigure}
        
            \caption{Variable gadget (a)~in the state \textsf{unset}
              as it appears in $\Gamma$ and $\Gamma'$, (b)~in the state
              \textsf{true}, (c)~in the state \textsf{false}.}
        \label{fig:variable-gadget}
    \end{figure}

	\paragraph{Split gadget.}
	A split gadget consists of a central vertex $c$ of degree~3
	together with paths to the left, right, and top;
	see \cref{fig:copy-gadget}.	
	They are enclosed in a system of tunnels formed by the forbidden area.
	\cref{fig:copy-gadget-pulled-in} shows a split gadget~$\sigma$
	as it appears in~$\Gamma$ and~$\Gamma'$.
	If the crossing/split/variable gadget to the left of~$\sigma$,
	which shares the vertex~$l$ with~$\sigma$,
	has a transferable
	vertex, then $l$ can be moved to the next corner of the tunnel.
	Then, in turn, $c$ can be moved to the bottom right corner
	of the white (obstacle-free) triangle in~$\sigma$
	(see \cref{fig:copy-gadget-pushed-out}),
	and the two other neighbors of~$c$ can be pushed one position
	to the right and one position up,
	allowing their neighbors (i.e., the (orange) vertices $o$ and $o'$, respectively)
	to move to the next corner of the tunnel.
	In this case we say that the horizontal (vertical) orientation of~$\sigma$
	is \textsf{right} (\textsf{top}) since $o$ ($o'$)
	is moved into the gadget to the top (right) of~$\sigma$.
	The vertex $o$ ($o'$) becomes a transferable vertex.
	
	Otherwise, the horizontal
	(vertical) orientation of~$\sigma$ is \textsf{left} (\textsf{bottom}).
	This corresponds to \textit{not} forwarding the value \textsf{true}
	to the gadget on the right (top).

	\begin{figure}[tbh]
		\centering
		\begin{subfigure}[t]{.38\linewidth}
			\centering
			\includegraphics[page=6]{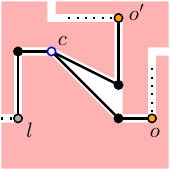}
			\caption{Split gadget as it appears in $\Gamma$ and $\Gamma'$ with
				horizontal and vertical orientation \textsf{left} and \textsf{bottom}.}
			\label{fig:copy-gadget-pulled-in}
		\end{subfigure}
		\hfill
		\begin{subfigure}[t]{.58\linewidth}
			\centering
			\includegraphics[page=7]{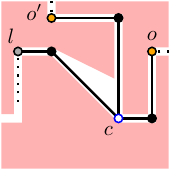}
			\caption{If the gadget on the left has horizontal
				orientation \textsf{right}, vertices can be pushed one position
				further and we can reach the orientations \textsf{right} and \textsf{top}.}
			\label{fig:copy-gadget-pushed-out}
		\end{subfigure}
		
		\caption{Split gadget: from a horizontal row transporting a truth value,
			we ``copy'' the same truth value up to a vertical column.}
		\label{fig:copy-gadget}
	\end{figure}

	\paragraph{Crossing gadget.}
	A crossing gadget has a similar structure as a split gadget.
	However, we now have a central vertex~$c'$ of degree~4
	with paths to the neighboring gadgets to the left, right, top, and bottom.
	In the center of the crossing gadget~$\chi$ in \cref{fig:crossing-gadget},
	\begin{figure}[tbh]
		\centering
		\begin{subfigure}{.24 \linewidth}
			\centering
			\includegraphics[page=8]{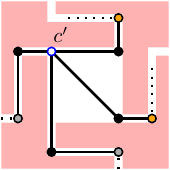}
			\caption{\textsf{left} + \textsf{bottom}}
			\label{fig:crossing-gadget-no-no}
		\end{subfigure}
		\hfill
		\begin{subfigure}{.24 \linewidth}
			\centering
			\includegraphics[page=9]{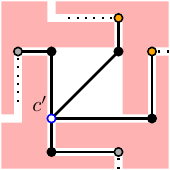}
			\caption{\textsf{right} + \textsf{bottom}}
			\label{fig:crossing-gadget-yes-no}
		\end{subfigure}
		\hfill
		\begin{subfigure}{.24 \linewidth}
			\centering
			\includegraphics[page=10]{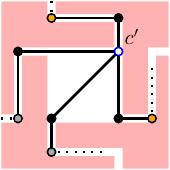}
			\caption{\textsf{left} + \textsf{top}}
			\label{fig:crossing-gadget-no-yes}
		\end{subfigure}
		\hfill
		\begin{subfigure}{.24 \linewidth}
			\centering
			\includegraphics[page=11]{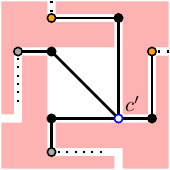}
			\caption{\textsf{right} + \textsf{top}}
			\label{fig:crossing-gadget-yes-yes}
		\end{subfigure}
		
		\caption{Crossing gadget: we transport truth values horizontally
			and vertically without influencing each other.
			The four possible combinations of orientations are illustrated.}
		\label{fig:crossing-gadget}
	\end{figure}
	there is a white (obstacle-free) square.
	In $\Gamma$ and $\Gamma'$, $c'$ is placed in the top left
	corner of this square (see \cref{fig:crossing-gadget-no-no}).
	If the gadget on the left of~$\chi$ has a transferable vertex,
	we can push the adjacent vertices to the next corners of the tunnel
	such that $c'$ can move to the bottom side of the square
	(see \cref{fig:crossing-gadget-yes-no,fig:crossing-gadget-yes-yes}).
	Only then can we push the vertices of the path leaving
	the gadget on the right side to the next corner of the tunnel.
	In this case, we have a transferable vertex, and we say
	that the horizontal orientation of the crossing gadget
	is \textsf{right}; otherwise it is \textsf{left}.
	Symmetrically, if the gadget below~$\chi$ has a transferable vertex,
	we can move $c'$ to the right side of the square
	(see \cref{fig:crossing-gadget-no-yes,fig:crossing-gadget-yes-yes}).
	Only then can we push the vertices of the path leaving~$\chi$
	through the top side to the next corner of the tunnel.
	In this case we have a transferable vertex and we say
	that the vertical orientation of the crossing gadget
	is \textsf{top}; otherwise it is \textsf{bottom}.
	Observe that the states of the gadgets on the left and below~$\chi$
	{\em independently} determine the horizontal and vertical
	orientation of~$\chi$.
	This property assures that we can transport information
	along routes that cross each other, but do not influence each other.
	We remark that not all crossing gadgets transmit information
	both to the right and to the top.
	If a crossing gadget is directly below a split gadget,
	then no information is pushed from the crossing to the split gadget
	as no vertex can be moved into a split gadget from below.
	The blue arrows in \cref{fig:hardness-sketch} indicate where
	information is pushed and where it is not.
	\Cref{fig:nph-full-construction} shows in more detail
	that there is no connection between a split gadget and
	the crossing gadget below.
	In such a case, the positions of the vertices in the crossing gadget that
	would lead to the gadget above are irrelevant for our reduction.

	\paragraph{Literal gadget.}
	\cref{fig:literal-gadget} shows a literal gadget~$\lambda$
	as it appears in~$\Gamma$ and~$\Gamma'$.
	It consists of five vertices, one of which,~$r$, is
	shared with the gadget~$\beta$ below; see \cref{fig:literal-gadget}.
	Only if the vertical orientation of~$\beta$ is \textsf{top},
	vertex~$r$ can move to the original position of~$s$ and~$s$ can move up.
	This in turn allows vertex~$t$ to move into the interior of~$\lambda$.
	In this case, we say that $\lambda$ has the state \textsf{true};
	otherwise, it has the state \textsf{false}.
	If a cycle in~$\Gamma$ or $\Gamma'$ contains obstacles,
	we call the cycle an \emph{anchor}.
	Observe that the anchor $\langle t, w, y \rangle$
	restricts the area where we can move~$t$.
	
	\begin{figure}[tbh]
		\centering
		\begin{subfigure}[t]{.41\linewidth}
			\centering
			\includegraphics[page=12]{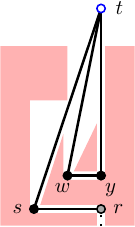}
			\caption{State \textsf{false} where vertex~$t$
				is above the forbidden area of the gadget.
				This is the literal gadget as it appears in~$\Gamma$
				and~$\Gamma'$.}
			\label{fig:literal-gadget-false}
		\end{subfigure}
		\hfill
		\begin{subfigure}[t]{.56\linewidth}
			\centering
			\includegraphics[page=13]{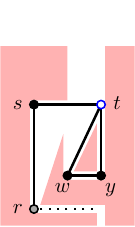}
			\caption{State \textsf{true} where vertex~$t$
				is inside the cavity formed by the forbidden area of the gadget.
				This state can only be reached if the split or crossing
				gadget below has orientation \textsf{top}.}
			\label{fig:literal-gadget-true}
		\end{subfigure}
		
		\caption{Literal gadget: Depending on the crossing or split gadget below, it can be in two different states.
			Only in the state \textsf{true}, vertex $t$ does not pop out of the gadget.}
		\label{fig:literal-gadget}
	\end{figure}
	
	\paragraph{Clause gadget.}
	Each clause~$c_i$ ($i \in [m]$) is represented by a clause gadget,
	which consists of a path of length~9 whose endpoints are
	anchored by two 3-cycles; see \cref{fig:clause-gadget}.
	In the clause gadget as it appears in~$\Gamma$ and~$\Gamma'$
	(see \cref{fig:clause-gadget-false}),
	we have exactly one free vertex (denoted by~$y$,
	or, to indicate that it belongs to clause~$c_i$, denoted by~$y_i$),
	which is at the bottom of a large rectangular obstacle-free region,
	which is also part of the synchronization gadget.
	Observe that, within this area, we cannot move $y$
	(up to a tiny bit) to the left or right
	due to its neighbors lying at (essentially) fixed positions
	(see \cref{fig:clause-gadget-false}).
	
	\begin{figure}[tbh]
		\centering
		\begin{subfigure}[t]{.47\linewidth}
			\centering
			\includegraphics[page=16]{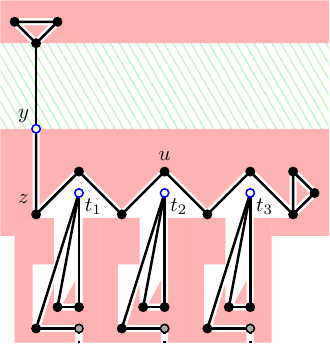}
			\caption{The clause gadget is in the state \textsf{false}.}
			\label{fig:clause-gadget-false}
		\end{subfigure}
		\hfill
		\begin{subfigure}[t]{.47\linewidth}
			\centering
			\includegraphics[page=17]{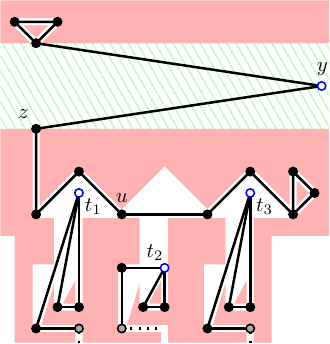}
			\caption{The clause gadget is in the state \textsf{true}.}
			\label{fig:clause-gadget-true}
		\end{subfigure}
		
		\caption{Clause gadget together with three literal gadgets:
			If at least one of the literal gadgets is in the state \textsf{true},
			the clause gadget is also in the state \textsf{true}.
			The obstacle-free region shared with the synchronization gadget is depicted in hatched green.}
		\label{fig:clause-gadget}
	\end{figure}
	
	However, if one of the incident literal gadgets is
	in the state \textsf{true},
	we can get a second free vertex by moving it onto the straight-line segment
	defined by its two neighbors (in \cref{fig:clause-gadget-false}, that
    vertex is called~$u$).
	Then, we can push, starting at that second free vertex
    (i.e., $u$ in \cref{fig:clause-gadget}), each vertex by
	one position to its successor until we have pushed~$z$ to the position
	where~$y$ is located in~$\Gamma$ and~$\Gamma'$
    (see \cref{fig:clause-gadget-true}).
	Now $y$ becomes a transferable vertex (we assume that we push~$y$
	out of the clause gadget into the synchronization gadget).
	We can move $y$ arbitrarily far to the right
	within this obstacle-free region (unless $y$ is blocked
	by the edges of another clause gadget).
	Only when this is done for all clause gadgets simultaneously,
	the synchronization gadget (see below) can be morphed as desired.
	Thus, we say that a clause gadget is in the state \textsf{true} if 
	at least one of its literal gadgets is in the state \textsf{true}; otherwise it is in the state \textsf{false}.

	\paragraph{Synchronization gadget.}
	The synchronization gadget is a 3-cycle $\langle v_1, v_2, v_3 \rangle$;
	see \cref{fig:nph-full-construction}.

	\begin{figure}[pt]
		\centering
		\includegraphics[page=18,scale=.95]{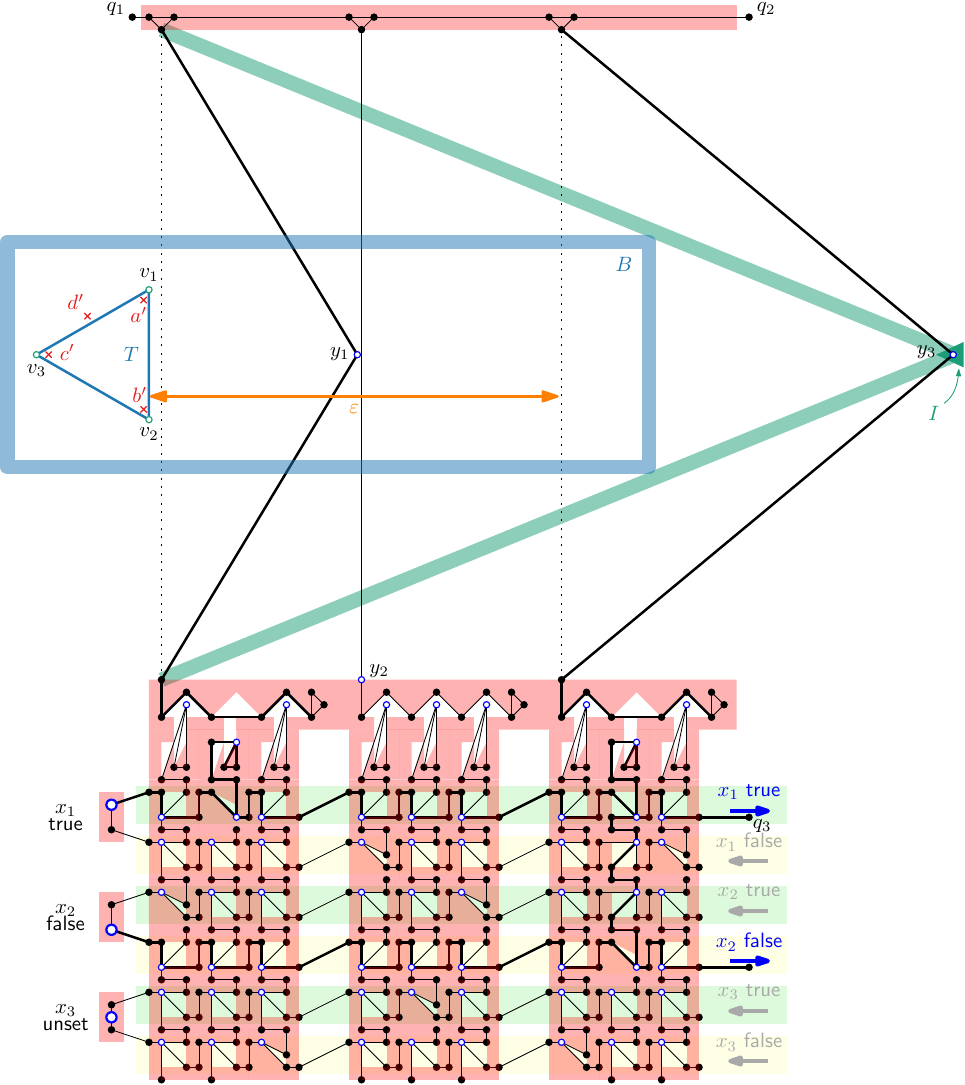}
		\caption{Full construction for %
			$\Phi =
			(x_2 \lor x_1\lor \neg x_3)\land
			(\neg x_1 \lor x_3 \lor x_2)\land
			(\neg x_3 \lor \neg x_2 \lor \neg x_1)$.
			Gadgets with orientation \textsf{right}/\textsf{top}/\textsf{true}
			use thicker strokes.
			The triangle $T$ and the height
			of the clause gadgets is scaled down;
			the rest is scaled up.
            The blue box~$B$ indicates the
            free area where $T$ can rotate.
            The green bars and~$I$ indicate
            how the edges incident to~$y$
            must bend to leave $B$ empty.
            The green and yellow horizontal strips
            contain the gadgets of the
            unnegated and negated literals, respectively.
        }
		\label{fig:nph-full-construction}
	\end{figure}

	\afterpage{\clearpage}

		In $\Gamma$,
	this cycle is drawn as a isosceles triangle $T$ with base $v_2v_3$
	that has~$v_1v_2$ as a vertical edge and with~$v_3$
	to the left of this edge (the exact shape or size of $T$ is not
	relevant).
	In $\Gamma'$, this cycle is drawn as a shifted version
	(see \cref{sec:intro} for the definition) of~$T$,
	which we call $T'$.
	Apart from this, $\Gamma$ and $\Gamma'$ are identical.
	The drawing $T$ is blocked from $T'$ due to \cref{prop:block-label-shift-C3-2}
	(we use the same construction as in the proof):
	inside of~$T$, we have obstacles $a'$, $b'$, and~$c'$
	arranged with $\varepsilon$ distance next to $v_1$, $v_2$,
	and $v_3$, respectively.
	Outside of~$T$ and with $\varepsilon$ distance
	to the midpoint of the edge $v_1v_3$,
	we have an obstacle~$d'$.
	According to \cref{prop:block-label-shift-C3-2},
	it suffices to set $\varepsilon$ to some constant at most $s/12$
	where $s$ is the length of the edge $v_2v_3$ in $\Gamma$.
	The only difference to the construction in \cref{prop:block-label-shift-C3-2}
	is that we do not have an obstacle $e'$,
	but we emulate the functionality of $e'$ by edges of
	the clause gadgets:
	$T$ and the obstacles~$a'$, $b'$, $c'$, and $d'$ are located in the large rectangular
	obstacle-free region of the clause gadget such that $v_1v_2$ lies directly to the left
	of the edge of the clause gadget that has $y_1$ as its lower endpoint
	(initially).
	We scale the rest of our construction such that the midpoint of the
	edge $v_1v_2$ is within $\varepsilon$ distance to all the edges of the clause gadget that have one of the vertices~$y_i$ as their lower endpoint
	(initially).\footnote{%
	Note that even if a clause gadget is in the state \textsf{false},
	we have a little flexibility in morphing the edges incident to~$y$
	due to the width of the tunnel in which the edge~$yz$ lies.
	When scaling the construction to have the edges incident to~$y$
	in $\varepsilon$ distance to the midpoint of $v_1v_2$,
	we account for this little extra slack.}
	We emphasize that (the coordinates of) $T$, the obstacles~$a'$, $b'$, $c'$, $d'$, and the distance~$\varepsilon$
	do not depend on~$\Phi$ (i.e., we can use the same $T$, $a'$, $b'$, $c'$, $d'$ for every given Boolean formula~$\Phi$).
	Hence, we can hard-code them into our reduction algorithm (the
	size of their encodings is constant);
	cf.\ \cref{prop:block-label-shift-C3-2}.
	By \cref{prop:block-label-shift-C3-one-outer},
	there exists an obstacle-avoiding planar morph~$M$
	from $T$ to~$T'$ with respect to the obstacle set formed by only the four obstacles $a'$, $b'$,
	$c'$, and $d'$.
	Let $B$ be an axis-aligned rectangle (or bounding box) that contains
	all drawings of~$M$ in its interior.
	The height of the large rectangular obstacle-free region of the
	synchronization gadget
	is chosen such that the region contains~$B$.
	Given that $T$, $T'$, $a'$, $b'$, $c'$, and $d'$ do not depend on~$\Phi$,
	it follows that $M$, and~$B$ also do not depend on~$\Phi$.
	Thus, the positions of the corners of~$B$ have constant-sized coordinates,
	which can again be hard-coded into our reduction algorithm.
	
	\paragraph{Correctness.}
    We first show a key property of the synchronization gadget.
    If $\Gamma$ can be morphed such that no edge or vertex except for~$T$
    lies inside $B$, then $T$ can be morphed to the shifted version~$T'$,
    and, after that, the first step of morphing $\Gamma$ can be reversed
    so that the final drawing coincides with~$\Gamma'$.
    Due to the placement of $T$, $a'$, $b'$, $c'$, and $d'$,
    this is possible if and only if
    all of the clause gadgets are in the state \textsf{true}:
    in this case, for each clause~$c_i$, we can move the vertex~$y_i$
    in the clause gadget of~$c_i$
    into a region~$I$ on the far right side
    so that the edges incident to~$y_i$ do not intersect $B$.
    In contrast, if a clause gadget is in the state \textsf{false},
    at least one of the edges incident to its vertex~$y_i$
    contains a point of distance at most~$\varepsilon$
    to the midpoint of $v_1v_2$,
    which means that we cannot morph $T$ to $T'$
    by \cref{prop:block-label-shift-C3-2}.
    
	Now, if~$\Phi$ admits a satisfying
	truth assignment, we can describe a morph from
	$\Gamma$ to $\Gamma'$ by moving the decision vertices in the variable
	gadgets according to the truth assignment,
	which allows to transport these truth values via the split and crossing
	gadgets to the literal and the clause gadgets.
	As all clause gadgets can reach the state \textsf{true} at the same time,
	we can morph the drawing of the 3-cycle in the synchronization gadget to its shifted version and then move all
	other vertices back to their original position.
	(Note that this morph can also be turned into a piecewise linear one by \cref{obs:piecewise}.)
	
	For the other direction, if a morph from $\Gamma$ to $\Gamma'$ exists,
	then we know by our construction that
	at some point, all clause gadgets were in the state \textsf{true} simultaneously (as this is necessary to morph the drawing of the 3-cycle in the synchronization gadget to its shifted version),
	which means the variable gadgets
	represent at the same time a satisfying truth assignment for~$\Phi$.

	\paragraph{Running time.}
	\label{par:number-obstacles}
	As stated above, we realize the forbidden areas in our gadgets by
	populating them with obstacles placed on a fine grid.
	For each gadget type $\mathcal T$ other than the synchronization gadget,
	there clearly exists a grid that is fine enough to ensure the desired
	functionality of the gadget. Essentially the required resolution of the grid
	is a function in the width of the tunnels used in $\mathcal T$. Notably,
	the design of $\mathcal T$ (in particular, the width of the tunnels and, thus, the
	resolution) does not depend on $\Phi$. Hence, the number of obstacles
	used in~$\mathcal T$ is $O(1)$ and there is a constant-time algorithmic
	subroutine that creates~$\mathcal T$ at some fixed location in the plane
	(for each gadget type~$\mathcal T$).
	By extension, there exists a polynomial-time algorithm
	(note that the overall number of used gadgets is $O(nm)$) that creates the
	desired grid-like arrangement of translated copies of gadgets that make up
	the \emph{bottom part} (everything below the synchronization gadget) of our
	construction.
	Regarding the synchronization gadget, recall that
	the coordinates of the vertices of $T$ and the obstacles $a'$, $b'$, $c'$, and $d'$,
	as well as the distance~$\varepsilon$ and the corners of~$B$ also do not
	depend
	on~$\Phi$ and can, thus, be harded-coded (determined in constant time)
	into the reduction algorithm.
	To the right side of the triangle~$T$,
	we slice a vertical strip of (constant) width $\varepsilon$
	into $m$ vertical substrips---one for each clause gadget.
	The bottom part of our construction is now scaled as stated in the
	description of the synchronization gadget, which decreases the
	the width and height of each gadget in the bottom part to $\Theta(1/m)$,
	which can again be done in polynomial time (recalling that the number
	of gadgets is $O(nm)$).

	\paragraph{Connectivity.}
	So far, the graph in our reduction has $\Omega(m)$ connected components:
	a large connected component comprising
	all variable, split, crossing and literal gadgets,
	a connected component for all clause gadgets,
	as well as another connected component for the synchronization gadget.
	We can merge these components by adding edges without influencing
	the behavior of our gadgets; see \cref{fig:nph-full-construction}:
	We can add a constant-length path between the vertices
	labeled~$v_1$ and~$q_1$ (this leaves sufficiently much freedom
	for $T$ to morph to $T'$) and another constant-length path
	between the vertices labeled~$q_2$ and~$q_3$ (near the upper
	blue arrow in \cref{fig:nph-full-construction}).
\end{proof}

\begin{remark}[Bounded number of obstacles per face]
	In our NP-hardness construction, even after connecting the graph,
	there are faces containing a large number of obstacles.
	Now one might wonder whether the problem becomes easier if we
	have a small constant number of obstacles per face.
	However, this is not the case because
	in a slightly adjusted NP-hardness construction,
	every face has at most one obstacle:
	enclose every obstacle of the current construction by
	a separate 3-cycle and connect this 3-cycle via a sufficiently
	long path to one of the vertices of the considered face.
	Now, clearly, every face has at most one obstacle and
	these new parts do not restrict the functionality
	of the gadgets because we have added them
	via sufficiently long paths.
\end{remark}

\section{Open Problems}
\label{sec:conclusion}

\begin{enumerate}
	\item In general, the decision problem whether, for two crossing-free
    straight-line drawings of the same graph, there exists an
    obstacle-avoiding crossing-free morph is NP-hard (\cref{thm:NP-hard*}).
    However, it can be solved efficiently for forests (\cref{lem:trees}).
	Are there other meaningful graph classes where this is the case?
	In particular, what about cycles or triangulations?
	It is conceivable that the latter case is actually easier since the
	placement of obstacles that are compatible with the two given drawings
	is quite limited.
	Regarding cycles, we emphasize that the existence of a
	free vertex is not a sufficient condition for the existence of
	a morph (cf. \cref{fig:dense-fox}).
	\item Does the problem lie in NP? Is it $\exists\mathbb{R}$-hard?
	\item Our NP-hardness reduction in the proof of
          \cref{thm:NP-hard*} uses a large number of obstacles.  On
          the other hand, \cref{prop:1or2-obstacles} says that every
          plane graph admits an obstacle-avoiding morph from one
          drawing to another if the number of obstacles (compatible
          with the two drawings) is at most~2.  So for how many
          obstacles does the problem become NP-hard?
	\item The drawings $\Gamma$ and $\Gamma'$ produced by our reduction are identical except for the position of three vertices.
	Does the problem become easier when only up to two vertices may change positions?
	\item It is easy to observe that if our reduction is applied to a satisfiable formula~$\Phi$
	with $n$ variables and $m$ clauses, there is a piecewise linear planar obstacle-avoiding morph between the produced drawings $\Gamma$ and $\Gamma'$ with $\Theta(n+m)$ steps, which is also necessary.
	Note that this number depends on the size of the output of the
        reduction.  This motivates the following family of questions.
	Let~$k$ be an arbitrary fixed constant.
	Given two planar straight-line drawings~$\Gamma$ and~$\Gamma'$ of the
	same plane graph and a set of obstacles compatible with~$\Gamma$
	and~$\Gamma'$, decide whether there exists a piecewise linear planar
	obstacle-avoiding morph from $\Gamma$ to $\Gamma'$ with at most $k$ steps.
	For which values of~$k$ can this decision problem be answered efficiently?
	\item Given two drawings of the same plane graph, how many compatible 
	obstacles are necessary and sufficient to block them?
	Can this be computed efficiently?
        
      \item Lubiw and Petrick \cite{lp-mpgdbe-JGAA11} studied the
        problem of morphing drawings of graphs where edges can be
        drawn as polygonal chains, that is, edges are allowed to bend.
        Does morphing in the presence of obstacles remain NP-hard if
        we allow a certain (say, constant) number of bends per edge?
        This can be seen as a special case of our problem where the
        input is restricted to graphs whose edges have been subdivided
        a certain number of times.
\end{enumerate}

\bibliographystyle{plainurl} %
\bibliography{cas-morphing}

\begin{thebibliography}{10}

\bibitem{aichholzer2011convexifying}
Oswin Aichholzer, Greg Aloupis, Erik~D. Demaine, Martin~L. Demaine, Vida
  Dujmovi{\'c}, Ferran Hurtado, Anna Lubiw, G{\"u}nter Rote, Andr{\'e} Schulz,
  Diane~L Souvaine, and Andrew Winslow.
\newblock Convexifying polygons without losing visibilities.
\newblock In Greg Aloupis and David Bremner, editors, {\em CCCG}, pages
  229--234, 2011.
\newblock URL: \url{http://www.cccg.ca/proceedings/2011/papers/paper70.pdf}.

\bibitem{DBLP:journals/siamcomp/AlamdariABCLBFH17}
Soroush Alamdari, Patrizio Angelini, Fidel Barrera{-}Cruz, Timothy~M. Chan,
  Giordano {Da Lozzo}, Giuseppe {Di Battista}, Fabrizio Frati, Penny Haxell,
  Anna Lubiw, Maurizio Patrignani, Vincenzo Roselli, Sahil Singla, and Bryan~T.
  Wilkinson.
\newblock How to morph planar graph drawings.
\newblock {\em {SIAM} J. Comput.}, 46(2):824--852, 2017.
\newblock \href {https://doi.org/10.1137/16M1069171}
  {\path{doi:10.1137/16M1069171}}.

\bibitem{alt2004complexity}
Helmut Alt, Christian Knauer, Günter Rote, and Sue Whitesides.
\newblock On the complexity of the linkage reconfiguration problem.
\newblock In {\em Towards a Theory of Geometric Graphs}, volume 342 of {\em
  Contemporary Mathematics}, pages 1--14. American Mathematical Society,
  Providence, 2004.

\bibitem{DBLP:conf/compgeom/AngeliniLFLPR15}
Patrizio Angelini, Giordano {Da Lozzo}, Fabrizio Frati, Anna Lubiw, Maurizio
  Patrignani, and Vincenzo Roselli.
\newblock Optimal morphs of convex drawings.
\newblock In Lars Arge and J{\'{a}}nos Pach, editors, {\em SoCG}, volume~34 of
  {\em LIPIcs}, pages 126--140. Schloss Dagstuhl~-- Leibniz-Zentrum f{\"{u}}r
  Informatik, 2015.
\newblock \href {https://doi.org/10.4230/LIPIcs.SOCG.2015.126}
  {\path{doi:10.4230/LIPIcs.SOCG.2015.126}}.

\bibitem{DBLP:journals/jgaa/ArsenevaBCDDFLT19}
Elena Arseneva, Prosenjit Bose, Pilar Cano, Anthony D'Angelo, Vida
  Dujmovi{\'c}, Fabrizio Frati, Stefan Langerman, and Alessandra Tappini.
\newblock Pole dancing: {3D} morphs for tree drawings.
\newblock {\em J. Graph Algorithms Appl.}, 23(3):579--602, 2019.
\newblock \href {https://doi.org/10.7155/jgaa.00503}
  {\path{doi:10.7155/jgaa.00503}}.

\bibitem{agi-mtds3-JGAA23}
Elena Arseneva, Rahul Gangopadhyay, and Aleksandra Istomina.
\newblock Morphing tree drawings in a small {3D} grid.
\newblock {\em J. Graph Algorithms Appl.}, 27(4):241--279, 2023.
\newblock \href {https://doi.org/10.7155/jgaa.00623}
  {\path{doi:10.7155/jgaa.00623}}.

\bibitem{befklow-mpgdt3d-CGT23}
Kevin Buchin, Will Evans, Fabrizio Frati, Irina Kostitsyna, Maarten
  L{\"o}ffler, Tim Ophelders, and Alexander Wolff.
\newblock Morphing planar graph drawings through {3D}.
\newblock {\em Comput. Geom. Topol.}, 2(1):5:1--5:18, 2023.
\newblock \href {https://doi.org/10.57717/cgt.v2i1.33}
  {\path{doi:10.57717/cgt.v2i1.33}}.

\bibitem{Cairns}
Stewart~S. Cairns.
\newblock Deformations of plane rectilinear complexes.
\newblock {\em Amer. Math. Monthly}, 51(5):247--252, 1944.
\newblock \href {https://doi.org/10.1080/00029890.1944.11999082}
  {\path{doi:10.1080/00029890.1944.11999082}}.

\bibitem{cm14}
{\'{E}}ric {Colin de Verdi{\`{e}}re} and Arnaud de~Mesmay.
\newblock Testing graph isotopy on surfaces.
\newblock {\em Discret. Comput. Geom.}, 51(1):171--206, 2014.
\newblock \href {https://doi.org/10.1007/s00454-013-9555-4}
  {\path{doi:10.1007/s00454-013-9555-4}}.

\bibitem{connelly2010locked}
Robert Connelly, Erik~D. Demaine, Martin~L. Demaine, Sándor~P. Fekete, Stefan
  Langerman, Joseph S.~B. Mitchell, Ares Ribó, and Günter Rote.
\newblock Locked and unlocked chains of planar shapes.
\newblock {\em Discrete Comput. Geom.}, 44(2):439--462, 2010.
\newblock \href {https://doi.org/10.1007/s00454-010-9253-3}
  {\path{doi:10.1007/s00454-010-9253-3}}.

\bibitem{Connelly}
Robert Connelly, Erik~D. Demaine, and G{\"u}nter Rote.
\newblock Straightening polygonal arcs and convexifying polygonal cycles.
\newblock {\em Discrete Comput. Geom.}, 30:205--239, 2003.
\newblock \href {https://doi.org/10.1007/s00454-003-0006-7}
  {\path{doi:10.1007/s00454-003-0006-7}}.

\bibitem{DBLP:journals/algorithmica/LozzoBFPR20}
Giordano {Da Lozzo}, Giuseppe {Di Battista}, Fabrizio Frati, Maurizio
  Patrignani, and Vincenzo Roselli.
\newblock Upward planar morphs.
\newblock {\em Algorithmica}, 82(10):2985--3017, 2020.
\newblock \href {https://doi.org/10.1007/s00453-020-00714-6}
  {\path{doi:10.1007/s00453-020-00714-6}}.

\bibitem{fhkkswz-mgdppo-sofsem24}
Oksana Firman, Tim Hegemann, Boris Klemz, Felix Klesen, Marie~Diana Sieper,
  Alexander Wolff, and Johannes Zink.
\newblock Morphing graph drawings in the presence of point obstacles.
\newblock In Henning Fernau, Serge Gaspers, and Ralf Klasing, editors, {\em
  {SOFSEM}}, volume 14519 of {\em LNCS}, pages 240--254. Springer, 2024.
\newblock \href {https://doi.org/10.1007/978-3-031-52113-3_17}
  {\path{doi:10.1007/978-3-031-52113-3_17}}.

\bibitem{gomes1999warping}
Jonas Gomes, Lucia Darsa, Bruno Costa, and Luiz Velho.
\newblock {\em Warping and Morphing of Graphical Objects}.
\newblock Morgan Kaufmann, 1999.

\bibitem{DBLP:journals/comgeo/KleistKLSSS19}
Linda Kleist, Boris Klemz, Anna Lubiw, Lena Schlipf, Frank Staals, and Darren
  Strash.
\newblock Convexity-increasing morphs of planar graphs.
\newblock {\em Comput. Geom.}, 84:69--88, 2019.
\newblock \href {https://doi.org/10.1016/j.comgeo.2019.07.007}
  {\path{doi:10.1016/j.comgeo.2019.07.007}}.

\bibitem{DBLP:conf/esa/Klemz21}
Boris Klemz.
\newblock Convex drawings of hierarchical graphs in linear time, with
  applications to planar graph morphing.
\newblock In Petra Mutzel, Rasmus Pagh, and Grzegorz Herman, editors, {\em
  ESA}, volume 204 of {\em LIPIcs}, pages 57:1--57:15. Schloss Dagstuhl~--
  Leibniz-Zentrum f{\"{u}}r Informatik, 2021.
\newblock \href {https://doi.org/10.4230/LIPIcs.ESA.2021.57}
  {\path{doi:10.4230/LIPIcs.ESA.2021.57}}.

\bibitem{lp-mpgdbe-JGAA11}
Anna Lubiw and Mark Petrick.
\newblock Morphing planar graph drawings with bent edges.
\newblock {\em J. Graph Algorithms Appl.}, 15(2):205--227, 2011.
\newblock \href {https://doi.org/10.7155/jgaa.00223}
  {\path{doi:10.7155/jgaa.00223}}.

\bibitem{DBLP:conf/gd/PurchaseHG06}
Helen~C. Purchase, Eve~E. Hoggan, and Carsten G{\"{o}}rg.
\newblock How important is the ``mental map''? -- {An} empirical investigation
  of a dynamic graph layout algorithm.
\newblock In Michael Kaufmann and Dorothea Wagner, editors, {\em GD}, volume
  4372 of {\em LNCS}, pages 184--195. Springer, 2006.
\newblock \href {https://doi.org/10.1007/978-3-540-70904-6\_19}
  {\path{doi:10.1007/978-3-540-70904-6\_19}}.

\bibitem{Thomassen}
Carsten Thomassen.
\newblock Deformations of plane graphs.
\newblock {\em J. Combin. Theory Ser. B}, 34(3):244--257, 1983.
\newblock \href {https://doi.org/10.1016/0095-8956(83)90038-2}
  {\path{doi:10.1016/0095-8956(83)90038-2}}.

\end{thebibliography}

\end{document}